\documentclass[12pt]{article}
\usepackage{amsmath, epsfig, cite}
\usepackage{amsthm}
\usepackage{amsfonts}
\usepackage{graphicx}
\usepackage{latexsym}
\usepackage{amssymb}
\usepackage{color}
\usepackage{url}
\usepackage{colortbl}
\usepackage{comment}
\usepackage[dvipsnames]{xcolor}

\usepackage{cite}
\usepackage{hyperref}
\usepackage{cleveref}

\usepackage{xfrac}
\usepackage{diagbox}

\DeclareMathAlphabet{\mathbfsl}{OT1}{ppl}{b}{it} 

\newcommand{\Z}{\mathbb{Z}}

\newcommand{\R}{\mathbb{R}}

\newcommand{\cA}{{\cal A}}

\newcommand{\cC}{{\cal C}}
\newcommand{\cD}{{\cal D}}

\newcommand{\cH}{{\cal H}}
\newcommand{\cI}{{\cal I}}

\newcommand{\cL}{{\cal L}}

\newcommand{\cS}{{\cal S}}

\newcommand{\bfa}{{\boldsymbol a}}
\newcommand{\bfb}{{\boldsymbol b}}
\newcommand{\bfc}{{\boldsymbol c}}

\newcommand{\bfw}{{\boldsymbol w}}
\newcommand{\bfx}{{\boldsymbol x}}
\newcommand{\bfy}{{\boldsymbol y}}
\newcommand{\bfz}{{\boldsymbol z}}

\DeclareMathOperator*{\argmax}{arg\,max}
\DeclareMathOperator*{\argmin}{arg\,min}

\newtheorem{theorem}{Theorem}
\newtheorem{lemma}{Lemma}

\newtheorem{corollary}{Corollary}

\newtheorem{example}{Example}

\theoremstyle{definition}
\newtheorem{definition}[theorem]{Definition}
\newtheorem{claim}[theorem]{Claim}

\begin{document}

\bibliographystyle{plain}

\title{On the Size of Balls and Anticodes of Small Diameter under the Fixed-Length Levenshtein Metric}

\author{
{\sc Daniella Bar-Lev}
 \hspace{1cm}
{\sc Tuvi Etzion}
 \hspace{1cm}
{\sc Eitan Yaakobi}
\thanks{The research of D. Bar-Lev was supported in part by the ISF grant no. 222/19. 
The research of T. Etzion was supported in part by the ISF grant no. 222/19 and by the Technion Data Science Initiative.
The research of E. Yaakobi was supported in part by the Israel Innovation Authority grant 75855 and the Technion Data Science Initiative.

An earlier version of this paper was presented in part at the 2021 IEEE International Symposium on Information Theory~\cite{BEY21}.

The authors are with the Department of Computer Science, Technion -- Israel Institute of Technology, Haifa 3200003, Israel, (e-mail: \{daniellalev,etzion,yaakobi\}@cs.technion.ac.il).}}

\maketitle

\begin{abstract}
The rapid development of DNA storage has brought the deletion and insertion channel to the front
line of research. When the number of deletions is equal to the number of insertions, the \emph{Fixed Length Levenshtein} ({FLL})
metric is the right measure for the distance between two words of the same length. Similar to any other metric, the size of a ball
is one of the most fundamental parameters.
In this work, we consider the minimum, maximum, and average size of a ball
with radius one, in the FLL metric. The related
minimum and the maximum size of a maximal anticode with diameter one are also considered.
\end{abstract}

\vspace{0.5cm}

\vspace{0.5cm}

\newpage

\section{Introduction}
Coding for DNA storage has attracted significant attention in the previous decade due to recent experiments and demonstrations
of the viability of storing information in macromolecules~\cite{Anavy19, BO21, CGK12, EZ17, Getal13, GH15,Oetal17,YGM17,TWAEHLSZM19}. Given the trends
in cost decreases of DNA synthesis and sequencing, it is estimated that already within this decade DNA storage may become
a highly competitive archiving technology. However, DNA molecules induce error patterns that are fundamentally different from their
digital counterparts~\cite{HMG18,HSR17, SOSAYY19, LSWY21}; This distinction results from the specific error behavior in DNA and it is well-known
that errors in DNA are typically in the form of substitutions, insertions, and deletions, where most published studies report that deletions
are the most prominent ones, depending upon the specific technology for synthesis and sequencing. Hence, due to its high relevance
to the error model in DNA storage coding for insertion and deletion errors has received renewed interest recently;
see e.g.~\cite{BGH17, BGZ16, Cheraghchi19, CK15, CS19,GW17, GS17, M09, MD06, RD14, SB19, SRB18, TPFV19}. This paper takes one more step in advancing this study and its goal is to study the size of balls
and anticodes when the number of insertions equals to the number of deletions.

If a word $\bfx \in \Z_q^n$ can be transferred to a word $\bfy \in \Z_q^n$ using $t$ deletions and $t$ insertions
(and cannot be transferred using a smaller number of deletions and insertions), then their
{\bf \emph{Fixed Length Levenshtein} (FLL) \emph{distance}} is $t$, which is denoted by $d_\ell (\bfx,\bfy) =t$.
It is relatively easy to verify that the FLL distance defines a metric.
Let $G=(V,E)$ be a graph whose set of vertices $V = \Z_q^n$
and two vertices $\bfx,\bfy \in V$ are connected by an edge if $d_\ell (\bfx,\bfy)=1$. This graph represents the FLL distance.
Moreover, the FLL distance defines a {\bf \emph{graphic metric}}, i.e., it is a metric and for each $\bfx,\bfy \in \Z_q^n$, $d_\ell (\bfx,\bfy)=t$
if and only if the length of the shortest path between $\bfx$ and $\bfy$ in $G$ is $t$.

One of the most fundamental parameters in any metric is the size of a ball with a given radius $t$ centered at a word~$\bfx$.
There are many metrics, e.g. the Hamming metric, the Johnson metric, or the Lee metric, where
the size of a ball does not depend on the word~$\bfx$. This is not the case in the FLL metric.
Moreover, the graph $G$ has a complex structure and it makes it much more difficult to find the
exact size of any ball and in particular the size of
a ball with minimum size and the size of a ball with maximum size. In~\cite{SaDo13}, a formula for the size of the ball with radius one, centered
at a word $x$, in the FLL metric was given. This formula depends on the number of runs in the word and the lengths of its alternating segments
(where in an alternating segment no run is larger than one). Nevertheless, while it is easy to compute the minimum size of a ball,
it is still difficult to determine from this formula what the maximum size of a ball is. In this paper, we find explicit expressions for
the minimum and maximum sizes of a ball when the ball is of radius one. We also find the average size of a ball when the radius
of the ball is one. Finally, we consider the related basic concept of anticode in the FLL metric, where an anticode with diameter $D$ is the
a code where the distance between any two elements of the code is at most $D$. Note, that a ball with radius $R$ has diameter
at most $2R$. We find the maximum size and the minimum size of maximal anticodes with diameter one, where an anticode with diameter one
is maximal if any addition of a word to it will increase its diameter.

This paper is the first one which considers a comprehensive discussion and exact computation on the balls with radius one and the
anticodes with diameter one in the FLL metric.
The rest of this paper is organized as follows. Section~\ref{cap: defenitions} introduces some basic concepts, presents some
of the known results on the sizes of balls, presents some results on equivalence of codes correcting deletions and insertions,
and finally introduce some observations required for our exposition.
The minimum size of a ball of any given radius in the FLL metric over $\Z_q$ is discussed in
Section~\ref{sec:min_size}. Section~\ref{sec:max_size} is devoted for the discussion on the maximum size of a ball with radius one
in the FLL metric over $\Z_q$. The analysis of non-binary sequences is discussed in Section~\ref{sec:max_non_binary}.
It appears that contrary to many other coding problems the binary case is much more difficult to analyze and it
is discussed in Section~\ref{sec:max_binary}. For the binary case,
the sequence for which the maximum size is obtained is presented in Theorem~\ref{the: q=2 max ball} and the maximum size is
given in Corollary~\ref{cor: q=2 max ball}. The average size of the FLL ball with radius one
over $\Z_q$ is computed in Section~\ref{sec:expect_size} and proved in Theorem~\ref{the: avg l-ball}.
In Section~\ref{sec:anticode_size}, we consider binary maximal anticodes with diameter one.
The maximum size of such an anticode is discussed in Section~\ref{sec:upper_anticodes}
and Section~\ref{sec:lower_anticodes} is devoted to the minimum size of such anticodes.
The results can be generalized for the non-binary case, but since they are more complicated and especially messy, they are omitted.

\section{Definitions and Previous Results}
\label{cap: defenitions}
In this section, we present the definitions and notations as well as several results that will be used throughout the paper.

For an integer $q\geq 2$, let $\Z_q$ denote the set of integers $\{0,1,\ldots,q-1\}$ and for an integer $n\ge0$,
let $\Z_q^n$ be the set of all sequences (words) of length $n$ over the alphabet $\Z_q$ and let $\Z_q^*=\bigcup_{n=0}^\infty\Z_q^n$, and let
$[n]$ denote the set of integers $\{1,2,\ldots,n\}$.
For two sequences $\bfx,\bfy\in\mathbb{Z}_q^n$, the distance between $\bfx$ and $\bfy$, $d(\bfx,\bfy)$, can be measured in various ways. When the type of errors is substitution, the \emph{Hamming distance} is the most natural to be considered.
The \emph{Hamming weight} of a sequence $\bfx\in\mathbb{Z}_q^*$, denoted by $\text{wt}{(\bfx})$,
is equal to the number of nonzero coordinates in $\bfx$.
The {Hamming distance} between two sequences ${\bfx,\bfy\in\mathbb{Z}_q^n}$, denoted by $d_H(\bfx,\bfy)$, is the number of coordinates
in which $\bfx$ and $\bfy$ differ. In other words, $d_H(\bfx,\bfy)$ is the number of symbol-substitution operations required
to transform $\bfx$ into $\bfy$. The Hamming distance is well known to be a metric on $\mathbb{Z}_q^n$
(also referred as the \emph{Hamming space}), as it satisfies  the three conditions of a metric (i.e., coincidence,
symmetry and the triangle inequality). Given a distance $d$ on a space $V$, the \emph{$t$-ball}
centered at ${\bfx \in V}$ is the set $\{ \bfy ~:~ d(\bfx,\bfy) \leq t\}$. The \emph{$t$-sphere} centered at ${\bfx \in V}$ is the set
$\{ \bfy ~:~ d(\bfx,\bfy) = t\}$. A \emph{code} $\cC \subseteq V$ is a subset of words from $V$.
The last related concept is an \emph{anticode} with diameter $D$ which is a code in $V$ for which the distance
between any two elements is at most $D$. Clearly, a $t$-ball is an anticode whose diameter is at most $2t$.
The \emph{Hamming $t$-ball} centered at ${\bfx\in\Z_q^n}$ will be denoted by $\cH_t(\bfx)$. For $\bfx\in\mathbb{Z}_q^n$, the number of words in the Hamming $t$-ball is a function of $n, q$ and $t$. The number of such words is
\begin{align}
\label{eq: hamming ball size}
|\cH_t(\bfx)|=\sum_{i=0}^t\binom{n}{i}(q-1)^i.
\end{align}

For an integer $t$, $0\le t\le n$, a sequence $\bfy\in\Z_q^{n-t}$ is a \emph{$t$-subsequence} of
$\bfx\in\Z_q^n$ if $\bfy$ can be obtained from $\bfx$ by deleting $t$ symbols from $\bfx$.
In other words, there exist $n-t$ indices ${1\le i_1<i_2<\cdots<i_{n-t}\le n}$ such that $y_j=x_{i_j}$, for all $1\le j\le n-t$.
We say that~$\bfy$ is a \emph{subsequence} of $\bfx$ if~$\bfy$ is a $t$-subsequence of $\bfx$ for some~$t$.
Similarly, a sequence $\bfy\in\Z_q^{n+t}$ is a \emph{$t$-supersequence} of~${\bfx\in\Z_m^n}$ if $\bfx$ is a $t$-subsequence
of~$\bfy$ and $\bfy$ is a \emph{supersequence} of $\bfx$ if $\bfy$ is a $t$-supersequence of $\bfx$ for some $t$.

\begin{definition}
The {\emph{deletion $t$-sphere}} centered at ${\bfx\in\Z_q^n}$, $\cD_t(\bfx)\subseteq \Z_q^{n-t}$, is the set of all $t$-subsequences of~$\bfx$.
The size of the largest deletion $t$-sphere in $\Z_q^n$ is denoted by $D_q(n,t)$. The {\emph{insertion $t$-sphere}} centered
at ${\bfx\in\Z_q^n}$, $\cI_t(\bfx)\subseteq \Z_q^{n+t}$, is the set of all $t$-supersequences of $\bfx$.
\end{definition}

Let $\bfx\in\mathbb{Z}_q^n$ be a sequence. The size of the insertion $t$-sphere $|\cI_t(\bfx)|$ does not depend on~$\bfx$ for any $0\le t\le n$. To be exact, it was shown by Levenshtein~\cite{L66} that
	\begin{align}~\label{eq: insertion ball size}
	|\cI_t(\bfx)|=\sum_{i=0}^t\binom{n+t}{i}(q-1)^i.
\end{align}
On the other hand, calculating the exact size of the deletion sphere is one of the more intriguing problems when studying codes for deletions. Deletion spheres, unlike substitutions balls and insertions spheres, are not \emph{regular}. That is, the size of
the deletion sphere, $|\cD_t(\bfx)|$, depends on the choice of the sequence $\bfx$.
Let $\{\sigma_1,\ldots, \sigma_q\}$ be the symbols of $\mathbb{Z}_q$ in some order and let $\bfc(n) = (c_1,c_2,\ldots, c_n)$ be a sequence in
$\mathbb{Z}_q^n$ such that $c_i = \sigma_i$ for $1 \leq i \leq q$ and $c_i=c_{i-q}$ for $i>q$.
It was shown in Hirschberg and Regnier~\cite{HR00} that $\bfc(n)$ has the largest deletion $t$-sphere and its size
is given by
\begin{align*}
D_q(n,t) = |\cD_t(\bfc(n))|=  \sum_{i=0}^t \binom{n-t}{i}D_{q-1}(t,t-i)
\end{align*}
In particular, $D_2(n,t) = \sum_{i=0}^t \binom{n-t}{i}$ and $D_3(n,t) = \sum_{i=0}^t \binom{n-t}{i}\sum_{j=0}^{t-i}\binom{i}{j}$.  The value $D_2(n,t)$ also satisfies the following recursion
$$D_2(n,t) = D_2(n-1,t) + D_2(n-2,t-1),$$
where the values for the basic cases can be evaluated by $D_2(n,t) = \sum_{i=0}^t \binom{n-t}{i}$.

\begin{definition}
A \emph{run} is a maximal subsequence composed of consecutive identical symbols. For a sequence  $\bfx\in\mathbb{Z}_q^n$,
the number of runs in $\bfx$ is denoted by $\rho(\bfx)$.
\end{definition}
\begin{example}\label{examp: runs}
If $\bfx=0000000$ then $\rho(\bfx)=1$ since $\bfx$ has a single run of length $7$ and for $\bfy=1120212$ we have that $\rho(\bfy) = 6$ since $\bfy$ has six runs, the first is on length two and the others are of length one.
\end{example}
There are upper and lower bounds on the size of the deletion ball which depend on the number of runs in the sequence.
Namely, it was shown by Levenshtein~\cite{L66} that
\begin{align*}
\binom{\rho(\bfx)-t+1}{t}\le |\cD_t(\bfx)|\le \binom{\rho(\bfx)+t-1}{t}.
\end{align*}
Later, the lower bound was improved in~\cite{HR00}:
\begin{align}
\label{eq: deletion ball size}
\sum_{i=0}^t \binom{\rho(\bfx)-t}{i} \leq |\cD_t(\bfx)| \leq \binom{\rho(\bfx)+t-1}{t}.
\end{align}
Several more results on this value which take into account the number of runs appear in~\cite{LL15}.

The \emph{Levenshtein distance} between two words $\bfx,\bfy \in \mathbb{Z}_q^*$, denoted by $d_L(\bfx,\bfy)$,
is the minimum number of insertions and deletions required to transform $\bfx$ into $\bfy$.
Similarly, for two sequences $\bfx,\bfy\in \mathbb{Z}_q^*$, $d_E(\bfx,\bfy)$ denotes the \emph{edit} distance between $\bfx$ and $\bfy$, which is the minimum number of insertions, deletions, and substitutions required to transform $\bfx$ into $\bfy$.
\begin{definition}
Let $t,n$ be integers such that  $0\le t\le n$. For a sequence $\bfx\in\mathbb{Z}_q^n$, the Levenshtein $t$-ball centered at ${\bfx\in\mathbb{Z}_q^n}$,
$\widehat{\cL}_t(\bfx)$, is defined by
$$
\widehat{\cL}_t(\bfx) \triangleq \{  \bfy\in\mathbb{Z}_q^* \ : \  d_L(\bfx,\bfy)\leq t \} .
$$
\end{definition}
In case $\bfx,\bfy\in\mathbb{Z}_q^n$, for some integer $n$, the \emph{Fixed Length Levenshtein} (FLL) \emph{distance} between $\bfx$ and $\bfy$,
$d_\ell(\bfx,\bfy)$, is the smallest $t$ for which there exists a $t$-subsequence $\bfz\in\mathbb{Z}_q^{n-t}$  of both $\bfx$ and $\bfy$, i.e.
\begin{equation}
\label{eq: deletion intersection}
d_{\ell}(\bfx,\bfy)= \min\{t': \cD_{t'}(\bfx)\cap \cD_{t'}(\bfy) \ne \varnothing\} = \frac{d_L(\bfx,\bfy)}{2}.
\end{equation}
In other words, $t$ is the smallest integer for which there exists $\bfz\in\mathbb{Z}_q^{n-t}$  such that
$\bfz\in \cD_t(\bfx)$ and $\bfy\in \cI_t(\bfz)$.
Note that if $\bfx,\bfy\in\mathbb{Z}_q^n$ and $\bfx$ is obtained from $\bfy$ by $t_1$ deletions and $t_2$ insertions, then $t_1=t_2$.


\begin{definition}
Let $n,t$ be integers such that $0\le t\le n$. For a sequence $\bfx\in\mathbb{Z}_q^n$, the FLL $t$-ball centered at ${\bfx\in\mathbb{Z}_q^n}$, $\cL_t(\bfx)\subseteq \mathbb{Z}_q^{n}$, is defined by
$$\cL_t(\bfx) \triangleq \{  \bfy\in\mathbb{Z}_q^n \ : \  d_\ell(\bfx,\bfy)\leq t \} .$$
\end{definition}

We say that a subsequence $\bfx_{[i,j]}\triangleq x_ix_{i+1}\cdots x_j$ is an \emph{alternating segment}
if $\bfx_{[i,j]}$ is a sequence of alternating distinct symbols $\sigma,\sigma'\in \Z_m$.
Note that $\bfx_{[i,j]}$ is a \emph{maximal alternating segment} if $\bfx_{[i,j]}$ is an alternating segment
and $\bfx_{[i-1,j]},\bfx_{[i,j+1]}$ are not. The number of maximal alternating segments of a sequence $\bfx$ will be denoted by $A(\bfx)$.
\begin{example}
If $\bfx=0000000$ then $A(\bfx)=7$ since $\bfx$ has seven maximal alternating segments, each of length one, and for $\bfx=1120212$ we have that $A(\bfx)=4$ and the maximal alternating segments are $1,\ 12,\ 202,\ 212$.
\end{example}

The following formula to compute $|\cL_1(\bfx)|$ as a function of $\rho(\bfx)$ and $A(\bfx)$ was given in~\cite{SaDo13}
\begin{align}
\label{eq:L1size}
\left|\cL_1(\bfx)\right| = \rho(\bfx)\cdot (n(q-1)-1) + 2 - \sum_{i=1}^{{\Large\text{$A$}}(\bfx)} \frac{(s_i-1)(s_i-2)}{2},
\end{align}
where $s_i$ for $1\le i\le A(\bfx)$ denotes the length of the $i$-th maximal alternating segment of $\bfx$.



Note that $|\widehat{\cL}_1(\bfx)|$, $|\widehat{\cL}_2(\bfx)|$ can be deduced from (\ref{eq: insertion ball size}), (\ref{eq: deletion ball size}), (\ref{eq: deletion intersection}), and  $|\cL_1(\bfx)|$, since
\begin{align*}
\widehat{\cL}_1(\bfx) & = \cD_1(\bfx)\cup \cI_1(\bfx)\cup\{\bfx\}, \\
\widehat{\cL}_2(\bfx) & = \cL_1(\bfx)\cup \cD_2(\bfx)\cup \cI_2(\bfx)\cup  \cD_1(\bfx)\cup \cI_1(\bfx),
\end{align*}
and the length of the sequences in each ball is different which implies that the sets in these unions are disjoint.
However, not much is known about the size of the Levenshtein ball and the FLL ball for arbitrary $n, t$ and $\bfx\in \mathbb{Z}_q^n$.

For $\bfx\in\mathbb{Z}_q^*$, let $|\bfx|$ denote the length of $\bfx$ and for a set of indices $I\subseteq [|\bfx|]$,
and let $\bfx_I$ denote the \emph{projection} of $\bfx$ on the ordered indices of $I$, which is the subsequence
of $\bfx$ received by the symbols in the entries of $I$.
For a symbol ${\sigma\in \Z_m}$, $\sigma^n$ denotes the sequence with $n$ consecutive $\sigma$'s.




A word $\bfx$ is called a \emph{common supersequence} (\emph{subsequence}) of some sequences $\bfy_1,\ldots,\bfy_t$ if $\bfx$ is a supersequence (subsequence) of each one of these $t$ words.
The set of all shortest common supersequences of $\bfy_1,\ldots,\bfy_t\in \mathbb{Z}_q^*$ is denoted by $\mathcal{SCS}(\bfy_1,\ldots,\bfy_t)$ and $\mathsf{SCS}(\bfy_1,\dots,\bfy_t)$ is the \emph{length of the shortest common supersequence} (\emph{SCS)} of $\bfy_1,\dots,\bfy_t$, that is,
$$\mathsf{SCS}(\bfy_1,\dots,\bfy_t) = \min_{\bfx\in \cS\cC\cS(\bfy_1,\ldots,\bfy_t)}\{|\bfx|\}.$$ Similarly, $\mathcal{LCS}(\bfy_1,\ldots,\bfy_t)$ is the set of all longest common subsequences of $\bfy_1,\dots,\bfy_t$ and $\mathsf{LCS}(\bfy_1,\dots,\bfy_t)$ is the \emph{length of the longest common subsequence} (\emph{LCS)} of $\bfy_1,\dots,\bfy_t$, that is,
$$
\mathsf{LCS}(\bfy_1,\dots,\bfy_t) \triangleq \max_{\bfx\in \cL\cC\cS(\bfy_1,\ldots,\bfy_t)}\{|\bfx|\}.
$$
This definition implies the following well known property.
\begin{claim}\label{lem: deletion intersection and LCS}
For $\bfx_1,\bfx_2\in\mathbb{Z}_q^n$, $\cD_t(\bfx_1)\cap \cD_t(\bfx_2)=\varnothing$ if and only if ${\mathsf{LCS}(\bfx_1,\bfx_2)< n-t}$.
\end{claim}

Combining (\ref{eq: deletion intersection}) and Claim~\ref{lem: deletion intersection and LCS} implies that
\begin{corollary}~\label{cor: LCS length}
If ${\bfx_1,\bfx_2\in\mathbb{Z}_q^n}$ then $$\mathsf{LCS}(\bfx_1,\bfx_2)= n-d_\ell(\bfx_1,\bfx_2).$$
\end{corollary}

For two sequences $\bfx\in \mathbb{Z}_q^{n}$ and $\bfy\in \mathbb{Z}_q^{m}$, the value of $\mathsf{LCS}(\bfx,\bfy)$  is given by the following recursive formula~\cite{Itoga81}

\begin{align}~\label{eq: recursive LCS}
\mathsf{LCS}(\bfx,\bfy)=
\begin{cases}
0 & n = 0 \text{ or } m = 0 \\
1 +\mathsf{LCS}( \bfx_{[1:{n}-1]}, \bfy_{[1:m-1]}) & x_{n}=y_{m}\\
\max \left\{
\mathsf{LCS}(\bfx_{[1:n-1]}, \bfy), \mathsf{LCS}(\bfx, \bfy_{[1:m-1]})
\right\} & \text{otherwise}
\end{cases}.
\end{align}

A subset $\cC\subseteq\mathbb{Z}_q^n$ is a \emph{$t$-deletion-correcting code} (\emph{${t\text{-insertion-correcting code}}$}, respectively) if for any two distinct codewords $\bfc,\bfc'\in\cC$ we have that $\cD_t(\bfc)\cap \cD_t(\bfc')=\varnothing$ (${\cI_t(\bfc)\cap \cI_t(\bfc')=\varnothing}$, respectively).
Similarly, $\cC$ is called a \emph{$(t_1,t_2)$-deletion-insertion-correcting code}
if for any two distinct codewords $\bfc,\bfc'\in\cC$ we have that
$\cD\cI_{t_1,t_2}(\bfc)\cap \cD\cI_{t_1,t_2}(\bfc')=\varnothing$, where $\cD\cI_{t_1,t_2}(\bfx)$ is the set of all words that can be obtained from $\bfx$ by $t_1$ deletions and $t_2$ insertions.
Levenshtein~\cite{L66} proved that $\cC$ is a $t$-deletion-correcting code if and only if $\cC$ is a $t$-insertion-correcting code and if and only if $\cC$ is a $(t_1,t_2)$-deletion-insertion-correcting code for every $t_1,t_2$ such that $t_1+t_2\le t$.
A straightforward generalization is the following result~\cite{CK13}.

\begin{lemma}
\label{lem: equivalent codes}
For all $t_1, t_2\in\Z$, if $\cC\subseteq\Z_q^n$ is a ${(t_1,t_2)\text{-deletion-insertion-correcting code}}$, then $\cC$ is also a $(t_1+t_2)$-deletion-correcting code.
\end{lemma}

\begin{corollary}
For $\mathcal{C}\subseteq \Z_q^n$, the following statements are equivalent.
\begin{enumerate}
\item $\cC$ is a $(t_1,t_2)$-deletion-insertion-correcting code.
\item $\cC$ is a $(t_1+t_2)$-deletion-correcting code.
\item $\cC$ is a $(t_1+t_2)$-insertion-correcting code.
\item $\cC$ is a $(t_1',t_2')$-deletion-insertion-correcting code for any $t_1',t_2'$ such that $t_1'+t_2' = t_1+t_2$.
\end{enumerate}
\end{corollary}

We further extend this result in the next lemma.
\begin{lemma}
A code $\cC\in\Z_q^n$ is a $(2t+1)$-deletion-correcting code if and only if the following two conditions are satisfied\\
$~~~\bullet$  $\cC$ is a $(t,t)$-deletion-insertion-correcting code \\
and also\\
$~~~\bullet$ if exactly $t+1$ FLL errors (i.e., $t+1$ insertions and $t+1$ deletions) occurred, then $\cC$ can detect these $t+1$ FLL errors.
\end{lemma}

\begin{proof}
If $\mathcal{C}$ is a $(2t+1)$-deletion-correcting code, then by definition for any $\bfc_1,\bfc_2\in \mathcal{C}$ we have that
$$
\cD_{2t+1}(\bfc_1)\cap \cD_{2t+1}(\bfc_2)=\varnothing.
$$
Therefore, by Claim~\ref{lem: deletion intersection and LCS} for any two distinct codewords $\bfc_1, \bfc_2\in \mathcal{C}$ we have that
$${\mathsf{LCS}(\bfc_1,\bfc_2)\le n-(2t+1)}.$$ Hence, by Corollary~\ref{cor: LCS length}, ${d_\ell(\bfc_1,\bfc_2)\ge 2(t+1)}$.
Since the FLL metric is graphic, it follows that  $\mathcal{C}$ can correct up to $t$ FLL errors
and if exactly $t+1$ FLL errors occurred it can detect them.

For the other direction, assume that $\mathcal{C}$  is a $(t,t)$-deletion-insertion-correcting code and if exactly $t+1$ FLL errors occurred, then $\cC$ can detect them. By Lemma~\ref{lem: equivalent codes}, $\mathcal{C}$
is a $(2t)$-deletion-correcting code which implies that ${\cD_{2t}(\bfc_1)\cap \cD_{2t}(\bfc_2) = \varnothing}$ for all $\bfc_1,\bfc_2\in\cC$, and hence by~(\ref{eq: deletion intersection}) we have that
$$
\forall \bfc_1,\bfc_2\in \mathcal{C}: \ \ \ d_\ell(\bfc_1,\bfc_2) > 2t.
$$
Let us assume to the contrary that there exist two codewords  $\bfc_1,\bfc_2\in \cC$ such that $d_\ell(\bfc_1,\bfc_2)=2t+1$.
Since the FLL metric is a graphic metric, it follows that there exists a word $\bfy\in\Z_q^n$ such that $d_\ell(\bfc_1,\bfy) = t$
and $d_\ell(\bfy,\bfc_2)= t+1$. Hence, if the received word is $\bfy$, then the submitted codeword can be either $\bfc_1$ ($t$ errors)
or $\bfc_1$ ($t+1$ errors) which contradicts the fact that in $\cC$ up to~$t$ FLL errors can be corrected
and exactly $t+1$ FLL errors can be detected. Hence,
$$
\forall \bfc_1,\bfc_2\in \cC: \ \ \ d_\ell(\bfc_1,\bfc_2) > 2t+1,
$$
and by definition, $\mathcal{C}$ can correct $2t+1$ deletions.
\end{proof}

\section{The Minimum Size of an FLL Ball}
\label{sec:min_size}

In this section, the explicit expression for the minimum size of an FLL $t$-ball of any radius $t$ is derived. Although this result is rather simple and straightforward, it is presented here for the completeness of the problems studied in the paper. Since changing the symbol in the $i$-th position from $\sigma$ to $\sigma'$ in any sequence $\bfx$ can be
done by first deleting $\sigma$ in the $i$-th position of $\bfx$ and then inserting $\sigma'$ in the same position of $\bfx$, it follows that
$$
\forall \bfx,\bfy\in\mathbb{Z}_q^n:\ \ \ d_H(\bfx,\bfy)\ge d_\ell(\bfx,\bfy).
$$
Since $\bfy\in \cH_t(\bfx)$ if and only if ${d_H(\bfx,\bfy)\le t}$ and
$\bfy\in \cL_t(\bfx)$ if and only if ${d_\ell(\bfx,\bfy)\le t}$, the following results are immediatey implied.

\begin{lemma}
\label{lem: hamming subset levinshtein balls}
If $n\ge t\ge0$ are integers and $\bfx\in\mathbb{Z}_q^n$, then $\cH_t(\bfx)\subseteq \cL_t(\bfx)$.
\end{lemma}
\begin{corollary}
\label{cor:BsubsetL}
For any two integers $n\ge t\ge 0$ and any sequence $\bfx\in\mathbb{Z}_q^n$, $|\cH_t(\bfx)|\le |\cL_t(\bfx)|$.
\end{corollary}

\begin{lemma}
\label{lem: hamming and ell minimal ball}
If $n>t\ge0$ are integers, then $\cH_t(\bfx) = \cL_t(\bfx)$ if and only if $\bfx=\sigma^n$ for $\sigma\in\mathbb{Z}_q$.
\end{lemma}
\begin{proof}
Assume first w.l.o.g. that $\bfx=0^n$ and let $\bfy\in \cL_t(\bfx)$ be a sequence obtained from $\bfx$ by at most~$t$ insertions and $t$ deletions. Hence, $\text{wt}(\bfy)\le t$ and $\bfy\in \cH_t(\bfx)$, which implies that ${\cL_t(\bfx)\subseteq \cH_t(\bfx)}$.
Therefore, Lemma~\ref{lem: hamming subset levinshtein balls} implies that $\cH_t(\bfx) = \cL_t(\bfx)$.

For the other direction, assume that $\cH_t(\bfx) = \cL_t(\bfx)$ and let $\bfx\in\Z_q^n$ were ${\bfx\ne \sigma^n}$ for all~$\sigma\in\Z_q$.
Since by Lemma~\ref{lem: hamming subset levinshtein balls}, $\cH_t(\bfx)\subseteq \cL_t(\bfx)$, to complete the proof, it is sufficient to show that there exists a sequence $\bfy\in \cL_t(\bfx)$\textbackslash $\cH_t(\bfx)$.
Denote $\bfx=(x_1,x_2,\ldots,x_n)$ and let $i$ be the smallest index for which $x_i\ne x_{i+1}$. Let $\bfy$ be the sequence defined by
$$ \bfy \triangleq \left(y_1,y_2,\ldots,y_{i-1},x_{i+1},x_{i}, y_{i+2},\ldots,y_{n}\right),$$
where $y_j\ne x_j$ for the first $t-1$ indices (for which ${j\notin\{ i,i+1\}}$) and $y_j=x_j$ otherwise.
Clearly, $\bfy$  differs from~$\bfx$ in $t+1$ indices and therefore $\bfy\notin \cH_t(\bfx)$.
On the other hand, $\bfy$ can be obtained from $\bfx$ by first deleting $x_i$ and inserting it to the right
of $x_{i+1}$ and then applying $t-1$ deletions and $t-1$ insertions whenever $y_j\ne x_j$ (where $j\notin\{i,i+1\}$).
Thus, $\bfy\in  \cL_t(\bfx)$\textbackslash $\cH_t(\bfx)$ which completes the proof.
\end{proof}

The following simple corollary is a direct result of Corollary~\ref{cor:BsubsetL}, Lemma~\ref{lem: hamming and ell minimal ball} and (\ref{eq: hamming ball size}). 
\begin{corollary}\label{cor: min l-ball}
If $n>t\ge 0$ and $m>1$ are integers, then the size of the minimum FLL $t$-ball is
$$\min_{\bfx\in\Z_q^n}\left|\cL_t(\bfx)\right| = \sum_{i=0}^t\binom{n}{i}(q-1)^i,$$
and the minimum is obtained only by the balls centered at $\bfx=\sigma^n$ for any $\sigma\in\Z_q$.
\end{corollary}

\section{The Maximum FLL Balls with Radius One}
\label{sec:max_size}

The goal of this section is to compute the size of a ball with maximum size and its centre.
For this purpose it is required first to compute the size of a ball.
The size of the FLL $1$-ball centered at $\bfx\in\mathbb{Z}_q^n$ was proved
in~\cite{SaDo13} and given in~(\ref{eq:L1size}).
In the analysis of the maximum ball we distinguish between the binary case and the non-binary case. Surprisingly, the computation of the non-binary case is not a generalization of the binary case. That is, the binary case is not a special case of the non-binary case.
Even more surprising is that the analysis of the non-binary case is much simpler than the analysis of the binary case.
Hence, we start with the analysis of the non-binary case which is relatively simple.

\subsection{The Non-Binary Case}
\label{sec:max_non_binary}


By (\ref{eq:L1size}), the size of a ball with radius one centered at $\bfx$ depends on $\rho(\bfx)$, the number of runs in~$\bfx$. For a given number of runs $1\le r\le n$, the size of a ball depends on the lengths of the maximal alternating segments in $\bfx$. The following lemma is an immediate consequence of (\ref{eq:L1size}).
\begin{lemma}
\label{cla:argminmax}
If $n>0$ and $1\le r\le n$, then
$$\argmax_{\substack{\bfx\in\mathbb{Z}_q^n\\ \rho(\bfx)=r}}|\cL_1(\bfx)| =  \argmin_{\substack{\bfx\in\mathbb{Z}_q^n\\ \rho(\bfx)=r}} \left\{\sum_{i=1}^{A(\bfx)} \frac{(s_i-1)(s_i-2)}{2}\right\}.$$
\end{lemma}

\begin{proof}
Let $\bfx\in\mathbb{Z}_q^n$ be a sequence with exactly $r$ runs.
Since $r (n(q-1)-1) + 2$ is a constant and
$$\sum_{i=1}^{A(\bfx)} \frac{(s_i-1)(s_i-2)}{2}\ge 0,$$
the claim follows immediately from (\ref{eq:L1size}).
\end{proof}

\begin{corollary}
\label{cor: max L1 for fix num of runs}
If $n>0$ and $1\le r\le n$, then
 $$\max_{\substack{\bfx\in\mathbb{Z}_q^n\\ \rho(\bfx)=r}}|\cL_1(\bfx)| = r(n(q-1)-1)+2 - \min_{\substack{\bfx\in\mathbb{Z}_q^n\\ \rho(\bfx)=r}} \left\{\sum_{i=1}^{A(\bfx)} \frac{(s_i-1)(s_i-2)}{2}\right\}.$$
\end{corollary}
Note that
\begin{align}
\label{obs:minsum}
\sum_{i=1}^{A(\bfx)} \frac{(s_i-1)(s_i-2)}{2} = 0 \iff \text{for each } 1\le i\le A(\bfx):\  s_i\in \{1,2\}.
\end{align}
The following claim is a straightforward result from the definitions of a run and an alternating segment.
\begin{lemma}
\label{cla:runandalt}
Let $n>0$ and let $\bfx\in\mathbb{Z}_q^n$ be a sequence. For $1\le i\le \rho(\bfx)$, denote by $r_i$ the length of the $i$-th run and by $\sigma_i\in \mathbb{Z}_q$ the symbol of the $i$-th run. Then all the maximal alternating segments of $\bfx$ have lengths at most two ($s_i\le 2$ for each $i$) if and only if for each $1\le i\le \rho(\bfx)-2$,  $\sigma_i\ne \sigma_{i+2}$ or  $r_{i+1}>1$.
\end{lemma}

The maximum value of $|{\cL_1(\bfx)}|$ for non-binary alphabet was given in~\cite{SGD14} without a proof.
For $q=2$ the value of $|{\cL_1(\bfx)}|$ given in~\cite{SGD14} without a proof is not accurate and we will give the
exact value with a complete proof.

\begin{theorem}
\label{the: maximal non-binary ell-ball}
The maximum FLL $1$-balls
 are the balls centered at ${\bfx\in\mathbb{Z}_q^n}$,
such that the number of runs in $\bfx$ is $n$ (i.e., any two consecutive symbols are different) and ${x_i\ne x_{i+2}}$  for all $1\le i\le n-2$.
In addition, the maximum size of an FLL $1$-ball is,
$$\max_{\bfx\in\Z_q^n}|{\cL_1(\bfx)}| = n^2(q-1) - n + 2.$$
\end{theorem}

\begin{proof}
Corollary~\ref{cor: max L1 for fix num of runs} implies that
\begin{footnotesize}
\begin{align*}
\max_{\bfx\in\mathbb{Z}_q^n}|\cL_1(\bfx)| & =
\max_{r\in\{1,\ldots, n\}}\left\{\max_{\substack{\bfx\in\mathbb{Z}_q^n \\ \rho(\bfx)=r}}|\cL_1(\bfx)|\right\}= \max_{r\in\{1,\ldots, n\}}\left\{ r(n(q-1)-1)+2 - \min_{\substack{\bfx\in\mathbb{Z}_q^n\\ \rho(\bfx)=r}} \left\{\sum_{i=1}^{A(\bfx)} \frac{(s_i-1)(s_i-2)}{2}\right\}\right\}.
\end{align*}\end{footnotesize}\\
Clearly, $r(n(q-1)-1)+2$ is maximized for $r=n$ and
therefore, using (\ref{obs:minsum}), we conclude that $\max_{\bfx\in\mathbb{Z}_q^n}|\cL_1(\bfx)|$ can be obtained for each $\bfx\in\mathbb{Z}_q^n$ such that $\rho(\bfx)=n$ and $s_i\le 2$ for each $i$. Note that $\sigma_i = x_i$ since $r=n$.
By Lemma~\ref{cla:runandalt}, it implies that $x_i\ne x_{i+2}$ or $r_{i+1}>1$ for each $1\le i\le n-2$.
Since $q>2$, it follows that there exists such an assignment for the symbols of each run such that $x_i\ne x_{i+2}$ for each $1\le i\le r-2$. It follows that
\begin{align*}
\max_{\bfx\in\mathbb{Z}_q^n}|\cL_1(\bfx)| & =   n^2(q-1)-n+2.
\end{align*}
\end{proof}

\subsection{The Binary Case}
\label{sec:max_binary}

The analysis to find the maximum ball for binary sequences is more difficult, since by definition of a run, there is no sequence $\bfx$ with $n$ runs such that $x_i\ne x_{i+2}$ (see Theorem~\ref{the: maximal non-binary ell-ball}) for some~$i$. Note also that since in the binary case two maximal alternating segments can not overlap it holds that $\sum_{i=1}^{A({\bfx})}s_i = n$ for any binary sequence $\bfx$.

For a sequence $\bfx\in\mathbb{Z}_2^n$, the \emph{alternating segments profile} of $\bfx$ is $(s_1,s_2,\ldots,s_{A(\bfx)})$. Note that each alternating segments profile defines exactly two binary sequences.

\begin{lemma}
\label{lem: q=2 runs and segments}
If $\bfx\in\mathbb{Z}_2^n$ then $\rho(\bfx) = n + 1 - A(\bfx)$.
\end{lemma}

\begin{proof}
Let $\bfx\in\mathbb{Z}_2^n$ be a sequence and let $\bfx_{[i,j]}$ and $\bfx_{[i',j']}$ be two consecutive maximal alternating segments such that $i <  i'$. Since $\bfx$ is a binary sequence, it follows that two maximal alternating segments cannot overlap, and hence $i'=j+1$.  Now, let $\alpha=A(\bfx)$ and we continue to prove the claim of the lemma by induction on $\alpha$ for any given $n\ge 1$. For $\alpha=1$, there is one maximal alternating segment whose length is clearly $n$ which consists of alternating symbols, i.e.,  there are $\rho(\bfx)=n$ runs as required. Assume the claim holds for any $\alpha'$ such that $1\le \alpha' <  \alpha$ and let $\bfx\in\mathbb{Z}_2^n$ be a sequence with exactly $\alpha$ maximal alternating segments.
Denote by  $\bfx'$ the sequence that is obtained from $\bfx$ by deleting its last maximal alternating  segment $\bfx''$.
By the induction hypothesis
$$\rho(\bfx')=(n-s_\alpha) + 1 - (\alpha-1) = n + 2 - s_\alpha  - t ,$$
where $s_\alpha$ is the length of $\bfx''$.
Clearly, the first symbol of $\bfx''$ is equal to the last symbol in $\bfx'$. Thus, 
$$\rho(\bfx) = \rho(\bfx'\bfx'') = \rho(\bfx') + s_\alpha - 1 = n + 2 - s_\alpha  - \alpha +s_\alpha - 1 = n + 1 - \alpha.$$
\end{proof}


Notice that Lemma~\ref{lem: q=2 runs and segments} does not hold for alphabet size $q> 2$.
To clarify, consider the sequences $\bfx_1 = 0120,\ \bfx_2 = 0101$ and $\bfx_3=0102 $, each of the sequences has four runs even though they differ in the number of maximal alternating segments; $A(\bfx_1) = 3,\  A(\bfx_2 ) = 1$ and $A(\bfx_3) = 2$.

\begin{definition}
For a  positive integer $\alpha$,  $\bfx^{(\alpha)}\in\mathbb{Z}_2^n$ is an {\bf\emph{$\alpha$-balanced sequence}} if ${A(\bfx)=\alpha}$
and ${s_i\in \{\lceil\frac{n}{\alpha}\rceil, \lceil\frac{n}{\alpha}\rceil - 1\}}$ for all $i\in\{1,\ldots,\alpha\}$. 
\end{definition}

\begin{lemma}
\label{lem: q=2 max ball for fix k}
If $n$ is a positive integer and $\alpha\in\{1,\ldots,n\}$ then
$$\argmax_{\substack{\bfx\in\mathbb{Z}_2^n \\ A(\bfx)=\alpha}}|\cL_1(\bfx)| = \left\{\bfx\in\mathbb{Z}_2^n:
\bfx\text{ is an } \alpha\text{-balanced sequence}  \right\}.$$
\end{lemma}

\begin{proof}
For a sequence $\bfx\in\mathbb{Z}_2^n$ such that $A(\bfx)=\alpha$, Lemma~\ref{lem: q=2 runs and segments} implies that $\rho(\bfx) =n+1-\alpha$.
Hence, by Lemma~\ref{cla:argminmax},
\begin{align*}
	\argmax_{\substack{\bfx\in\mathbb{Z}_2^n\\ A(\bfx)=\alpha}}|\cL_1(\bfx)| & = \argmin_{\substack{\bfx\in\mathbb{Z}_2^n\\ A(\bfx)=\alpha}} \sum_{i=1}^{\alpha} \frac{(s_i-1)(s_i-2)}{2}  \\ &
	= \argmin_{\substack{\bfx\in\mathbb{Z}_2^n\\ A(\bfx)=\alpha}}
	\sum_{i=1}^\alpha (s_i^2 - 3s_i + 2) \\
	& = \argmin_{\substack{\bfx\in\mathbb{Z}_2^n\\ A(\bfx)=\alpha}}\left(\sum_{i=1}^\alpha s_i^2 - 3 \sum_{i=1}^\alpha s_i + 2\alpha\right) \\ &
	 \stackrel{{(a)}}{=} \argmin_{\substack{\bfx\in\mathbb{Z}_2^n\\ \alpha(\bfx)=t}}\left(\sum_{i=1}^t s_i^2 - 3n + 2\alpha \right) \\
	& = \argmin_{\substack{\bfx\in\mathbb{Z}_2^n\\ A(\bfx)=\alpha}}\sum_{i=1}^\alpha s_i^2,
\end{align*}
where $(a)$ holds since alternating segments cannot overlap for binary sequences and therefore ${\sum_{i=1}^\alpha s_i=n}$.

Assume $\bfx\in\mathbb{Z}_2^n$ is a sequence such that $A(\bfx)=\alpha$, $(s_1,\ldots,s_\alpha)$ is the alternating segments profile of $\bfx$ and $\sum_{i=1}^\alpha s_i^2$ is minimal among all sequences in $\mathbb{Z}_2^n$. Assume to the contrary that $\bfx$ is not an $\alpha$-balanced sequence. Then there exist indices $i\ne j$
such that $s_i\le \left\lceil\frac{n}{\alpha}\right\rceil -1$ and $s_j>\left\lceil\frac{n}{\alpha}\right\rceil$ or there exist indices $i\ne j$
such that $s_i< \left\lceil\frac{n}{\alpha}\right\rceil -1$ and $s_j\ge \left\lceil\frac{n}{\alpha}\right\rceil$. Consider a sequence $\bfx'$ with the alternating segments profile $(\nu_1,\ldots,\nu_\alpha)$ where $$\nu_k=\begin{cases} s_i + 1 & \text{if } k=i\\
				 s_j - 1 & \text{if } k=j\\
				 s_k    & \text{otherwise}.
\end{cases}$$
Therefore,
\begin{align*}
	\sum_{k=1}^\alpha \nu_k^2 - \sum_{k=1}^\alpha s_k^2 & = \sum_{k=1}^\alpha \left( \nu_k^2 - s_k^2  \right) =  (\nu_i^2 - s_i^2) + (\nu_j^2 - s_j^2) \\
	& =\left((s_i+1)^2-s_i^2\right) +  \left((s_j-1)^2-s_j^2\right)  \\
	& = \left(s_i^2+2s_i+1-s_i^2\right) + \left(s_j^2-2s_j+1-s_j^2\right) \\
	& = 2(s_i - s_j +1) \\
	& < 2\left(\left\lceil\frac{n}{\alpha}\right\rceil -1-\left\lceil\frac{n}{\alpha}\right\rceil +1\right) = 0,
\end{align*}
and hence $\sum_{k=1}^\alpha \nu_k^2 < \sum_{k=1}^\alpha s_k^2$. This implies that if $\bfx$ is not an $\alpha$-balanced sequence, then $ \sum_{k=1}^\alpha s_k^2$ is not minimal, a contradiction.
Thus,
$$\argmax_{\substack{\bfx\in\mathbb{Z}_2^n\\ A(\bfx)=\alpha}}|\cL_1(\bfx)| =  \argmin_{\substack{\bfx\in\mathbb{Z}_2^n\\ A(\bfx)=\alpha}}\sum_{i=1}^\alpha s_i^2
 = \left\{\bfx\in\mathbb{Z}_2^n\ :\  \bfx\text{ is an } \alpha\text{-balanced sequence}  \right\}.$$
\end{proof}

\begin{lemma}
\label{lem: k balanced ball size}
Let $\bfx^{(\alpha)}$ be an $\alpha$-balanced sequence of length $n$. Then,
\begin{small}
\begin{align*}
\left|\cL_1\left(\bfx^{(\alpha)}\right)\right| & = (n+1-\alpha)(n-1) +2
 - \frac{k}{2}\left(\left\lceil\frac{n}{\alpha}\right\rceil-1\right)
\left(\left\lceil\frac{n}{\alpha}\right\rceil-2\right)
 - \frac{\alpha-k}{2}\left(\left\lceil\frac{n}{\alpha}\right\rceil-2\right)
 \left(\left\lceil\frac{n}{\alpha}\right\rceil-3\right),
 \end{align*}
 \end{small}
where $k\equiv n\pmod \alpha$ and $1\le k\le \alpha$.
\end{lemma}

\begin{proof}
By (\ref{eq:L1size}) we have that
\begin{align}
\label{eq:L1ofxk}
\left|\cL_1\left(\bfx^{(\alpha)}\right)\right| = \rho\left(\bfx^{(\alpha)}\right)\cdot (n-1)+2 - \sum_{i=1}^{\alpha}\frac{(s_i-1)(s_i-2)}{2},
\end{align}
and Lemma~\ref{lem: q=2 runs and segments} implies that
$ \rho\left(\bfx^{(\alpha)}\right)=n+1-\alpha$.
Let $k$ be the number of entries in the alternating segments profile of $\bfx^{(\alpha)}$ such that $s_i = \lceil\frac{n}{\alpha}\rceil$. Note forther that $\sum_{i=1}^\alpha s_i = n$ and $s_i\in\{\lceil\frac{n}{\alpha}\rceil,\lceil\frac{n}{\alpha}\rceil-1 \}$ for $1\le i\le \alpha$. Hence,
$$k \left\lceil\frac{n}{\alpha}\right\rceil + (\alpha-k) \left(\left\lceil\frac{n}{\alpha}\right\rceil-1\right) = n,$$
which is equivalent to
$$k = n- \alpha\left(\left\lceil\frac{n}{\alpha}\right\rceil - 1\right).$$
Therefore, $k$ is the value between $1$ to $\alpha$ such that $k\equiv n\pmod \alpha$. Thus, by (\ref{eq:L1ofxk}) we have that

\begin{small}
\begin{align*}
\left|\cL_1\left(\bfx^{(\alpha)}\right)\right| & = (n+1-\alpha)(n-1) +2
 - \frac{k}{2}\left(\left\lceil\frac{n}{\alpha}\right\rceil-1\right)
\left(\left\lceil\frac{n}{\alpha}\right\rceil-2\right)
 - \frac{\alpha-k}{2}\left(\left\lceil\frac{n}{\alpha}\right\rceil-2\right)
 \left(\left\lceil\frac{n}{\alpha}\right\rceil-3\right).
 \end{align*}
 \end{small}\end{proof}
By Lemma~\ref{lem: q=2 max ball for fix k} we have that
\begin{align*}
    \max_{x\in\mathbb{Z}_2^n}|\cL_1(\bfx)| &
    = \max_{1\le \alpha\le n}\left\{ \max_{\substack{\bfx\in\mathbb{Z}_2^n \\ A(\bfx)=\alpha}}|\cL_1(\bfx)|\right\}
    = \max_{1\le \alpha\le n}\left\{\left|\cL_1\left(\bfx^{(\alpha)}\right)\right|\right\},
\end{align*}
\noindent
and the size $\left|\cL_1\left(\bfx^{(\alpha)}\right)\right|$  for $1\le \alpha\le n$ is given in Lemma~\ref{lem: k balanced ball size}.
Hence, our goal is to find the set
$$\mathsf{A} \triangleq  \argmax_{1\le \alpha\le n}\left\{\left|\cL_1\left(\bfx^{(\alpha)}\right)\right|\right\},$$
i.e., for which values of $\alpha$ the maximum of $|\cL_1\left(\bfx^{(\alpha)} \right)|$ is obtained. The answer for this question is given in the following lemma whose proof can be found in the Appendix.


\begin{lemma}
\label{lam : max t}
Let $\bfx^{(\alpha)}$ be an $\alpha$-balanced sequence of length $n>1$. Then,
 $$\left|\cL_1\left(\bfx^{(\alpha)}\right)\right|> \left|\cL_1\left(\bfx^{(\alpha-1)}\right)\right|$$
if and only if  $n>2(\alpha-1)\alpha$.
\end{lemma}
\begin{theorem}
\label{the: q=2 max ball}
If $n$ is an integer, then
$$\mathsf{A} = \argmin_{\alpha\in\mathbb{N}}\left\{\left|\alpha-\frac{1}{2}\sqrt{1+2n}\right|\right\},$$
and the maximum FLL $1$-balls are the balls centered at the $\alpha$-balanced sequences of length $n$,
for $\alpha\in\mathsf{A}$. In addition, the size of the maximum FLL $1$-balls is given by
\begin{small}
\begin{align*}
& \max_{\bfx\in\mathbb{Z}_2^n}  \left\{|\cL_1(\bfx)|\right\}  = n^2 -n\alpha +\alpha+ 1 - \frac{k}{2}\left(\left\lceil\frac{n}{\alpha}\right\rceil-1\right)
\left(\left\lceil\frac{n}{\alpha}\right\rceil-2\right)   - \frac{\alpha-k}{2}\left(\left\lceil\frac{n}{\alpha}\right\rceil-2\right)
\left(\left\lceil\frac{n}{\alpha}\right\rceil-3\right)
,
\end{align*}
\end{small}
where $k\equiv n\pmod \alpha$ and $1\le k\le \alpha$.
\end{theorem}

\begin{proof}
Let $n$ be a positive integer. By Lemma~\ref{lem: q=2 max ball for fix k} we have that
\begin{align*}
    \max_{x\in\mathbb{Z}_2^n}|\cL_1(\bfx)| &
    = \max_{1\le \alpha\le n}\left\{ \max_{\substack{\bfx\in\mathbb{Z}_2^n \\ A(\bfx)=\alpha}}|\cL_1(\bfx)|\right\}
    = \max_{1\le \alpha\le n}\left\{\left|\cL_1\left(\bfx^{(\alpha)}\right)\right|\right\}.
\end{align*}
If there exists an integer $\alpha$, $1\le \alpha\le n$ such that  $n=2(\alpha-1)\alpha$, then by Lemma~\ref{lem: k balanced ball size}, $\left|\cL_1\left(\bfx^{(\alpha)}\right)\right|=\left|\cL_1\left(\bfx^{(\alpha-1)}\right)\right|$.
Additionally, by Lemma~\ref{lam : max t} we have that $\left|\cL_1\left(\bfx^{(\alpha)}\right)\right|>\left|\cL_1\left(\bfx^{(\alpha-1)}\right)\right|$ for $n>2(\alpha-1)\alpha$ which implies that $\left|\cL_1\left(\bfx^{(\alpha)}\right)\right|$ is maximized for $\alpha\in\{1,\ldots,n\}$ such that
\begin{align}\label{eq: alpha max ball}
    	2\alpha\left(\alpha+1\right) \ge n \ge 2\left(\alpha-1\right)\alpha.
\end{align}
To find $\alpha$ we have to solve the two quadratic equations from (\ref{eq: alpha max ball}). The solution for $\alpha$ must satisfies both equations and hence
$- \frac{1}{2} + \frac{\sqrt{1+2n}}{2}\le \alpha\le \frac{1}{2} + \frac{\sqrt{1+2n}}{2}$.
Namely, for $\alpha\in\mathsf{A}$,
$$\max_{\bfx\in\mathbb{Z}_2^n}\left\{|\cL_1(\bfx)|\right\} = \left|\cL_1\left(\bfx^{(\alpha)}\right)\right|$$
The size of $\cL_1\left(\bfx^{(\alpha)}\right)$ was derived in Lemma~\ref{lem: k balanced ball size}, which completes the proof.
\end{proof}

\begin{corollary}
\label{cor: q=2 max ball}
Let $n$ be an integer. Assuming $n$ is sufficiently large, we have that
$$\max_{\bfx\in\mathbb{Z}_2^n}\left\{|\cL_1(\bfx)|\right\} =  n^2 - \sqrt{2}n^{\frac{3}{2}}+O(n).$$
\end{corollary}
\begin{proof}
By Theorem~\ref{the: q=2 max ball} we have that
$\max_{\bfx\in\mathbb{Z}_2^n}\left\{|\cL_1(\bfx)|\right\} = \left|\cL_1\left(\bfx^{(\alpha)}\right)\right|$ for $\alpha=\left[\frac{1}{2}\sqrt{1+2n}\right]$. By Lemma~\ref{lem: k balanced ball size} we have that
\begin{small}
\begin{align*}
\left|\cL_1\left(\bfx^{(\alpha)}\right)\right| & = (n+1-\alpha)(n-1) +2
 - \frac{k}{2}\left(\left\lceil\frac{n}{\alpha}\right\rceil-1\right)
\left(\left\lceil\frac{n}{\alpha}\right\rceil-2\right)
 - \frac{\alpha-k}{2}\left(\left\lceil\frac{n}{\alpha}\right\rceil-2\right)
 \left(\left\lceil\frac{n}{\alpha}\right\rceil-3\right).
 \end{align*}
 \end{small}
Notice that
$$\frac{1}{2}\left(\sqrt{1+2n}-2\right)\le \alpha\le \frac{1}{2}\left(\sqrt{1+2n}+2\right)$$
and hence,
$\alpha = \frac{\sqrt{1+2n}}{2} + \epsilon_1$, where $|\epsilon_1|\le 1$.
Similarly,
$$
\frac{2n}{\sqrt{1+2n}+2}\le \left\lceil\frac{2n}{\sqrt{1+2n}+2}\right\rceil\le \left\lceil\frac{n}{\alpha}\right\rceil \le \left\lceil\frac{2n}{\sqrt{1+2n}-2}\right\rceil\le \frac{2n}{\sqrt{1+2n}-2}+1.
$$
which implies that
$$
\left\lceil\frac{n}{\alpha}\right\rceil = \frac{2n}{\sqrt{1+2n}} + \epsilon_2,
$$
where by simple calculation we can find that $|\epsilon_2|\le 3$. Thus,
\begin{align*}
\max_{\bfx\in\mathbb{Z}_2^n}& |\cL_1(\bfx)| =  (n+1-\alpha)(n-1) +2
 - \frac{k}{2}\left(\left\lceil\frac{n}{\alpha}\right\rceil-1\right)
\left(\left\lceil\frac{n}{\alpha}\right\rceil-2\right)
 - \frac{\alpha-k}{2}\left(\left\lceil\frac{n}{\alpha}\right\rceil-2\right)
 \left(\left\lceil\frac{n}{\alpha}\right\rceil-3\right) \\
 & = (n+1-\alpha)(n-1) +2 - \frac{k}{2}\left(\left\lceil\frac{n}{\alpha}\right\rceil-2\right) \left(\left\lceil\frac{n}{\alpha}\right\rceil-1 -\left\lceil\frac{n}{\alpha}\right\rceil+3 \right)
 -\frac{\alpha}{2} \left(\left\lceil\frac{n}{\alpha}\right\rceil-2\right)
 \left(\left\lceil\frac{n}{\alpha}\right\rceil-3\right)  \\
 & =  (n+1-\alpha)(n-1) +2
  - k\left(\left\lceil\frac{n}{\alpha}\right\rceil-2\right)  -\frac{\alpha}{2} \left(\left\lceil\frac{n}{\alpha}\right\rceil-2\right)
 \left(\left\lceil\frac{n}{\alpha}\right\rceil-3\right)   \\
 & =  (n+1-\frac{\sqrt{1+2n}}{2} - \epsilon_1)(n-1) +2
  - k\left(\frac{2n}{\sqrt{1+2n}} + \epsilon_2-2\right) \\
  & \ \ \  -\frac{{\sqrt{1+2n}} + 2\epsilon_1}{4} \left(\frac{2n}{\sqrt{1+2n}} + \epsilon_2-2\right)
 \left(\frac{2n}{\sqrt{1+2n}} + \epsilon_2-3\right) \\
 & = n^2 +1 -\left(\frac{\sqrt{1+2n}}{2} + \epsilon_1\right)(n-1) \\
 & \ \ \ - \left(\frac{2n}{\sqrt{1+2n}} + \epsilon_2-2\right)\left(k + \frac{{\sqrt{1+2n}} + 2\epsilon_1}{4}
 \left(\frac{2n}{\sqrt{1+2n}} + \epsilon_2-3\right) \right).
  \end{align*}
  Note that $1\le k \le \alpha\le \frac{1}{2}\left(\sqrt{1+2n}+2\right)$, which implies that
  \begin{align*}
  \max_{\bfx\in\mathbb{Z}_2^n} |\cL_1(\bfx)| 
 & = n^2 - \frac{n\sqrt{1+2n}}{2} - \frac{n^2}{\sqrt{1+2n}} + O(n)\\
 & = n^2 - \sqrt{2}n^{\frac{3}{2}}+O(n).
  \end{align*}
\end{proof}

\section{The Expected Size of an FLL $1$-Ball}
\label{sec:expect_size}

Let $n$ and $q>1$ be integers and let $\bfx\in\mathbb{Z}_q^n$ be a sequence.
By (\ref{eq:L1size}), for every $\bfx\in\mathbb{Z}_q^n$, we have
\begin{align*}
 |\cL_1(\bfx)|& = \rho(\bfx) (n(q-1)-1) + 2 - \sum_{i=1}^{A(\bfx)} \frac{(s_i-1)(s_i-2)}{2}
	& 
	\\ & =  \rho(\bfx)(nq-n-1) + 2 -\frac{1}{2} \sum_{i=1}^{A(\bfx)} s_i^2 + \frac{3}{2} \sum_{i=1}^{A(\bfx)} s_i-  A(\bfx).
 \end{align*}
Thus, the average size of an FLL $1$-ball is
\begin{align}
\label{eq: avg L1}
   	\mathop{{}\mathbb{E}}_{\bfx\in\mathbb{Z}_q^n}\left[\left|\cL_1(\bfx)\right|\right]
	& = \mathop{{}\mathbb{E}}_{\bfx\in\mathbb{Z}_q^n}\left[\rho(\bfx)(n(q-1)-1) + 2
	-\frac{1}{2} \sum_{i=1}^{A(\bfx)} s_i^2 + \frac{3}{2} \sum_{i=1}^{A(\bfx)} s_i-  A(\bfx)\right].
\end{align}

\begin{lemma}
\label{lam: avg sum si}
For any two integers $n, q>1$,
$$\mathop{{}\mathbb{E}}_{\bfx\in\mathbb{Z}_q^n}\left[\sum_{i=1}^{A(\bfx)}s_i\right] = n + (n-2)\cdot \frac{(q-1)(q-2)}{q^2}.$$
\end{lemma}

\begin{proof}
If $\bfx\in\mathbb{Z}_q^n$, then by the definition of an alternating segment, we have that for each ${1\le i\le n}$, $x_i$ is contained in at
least one maximal alternating segment and not more than two maximal alternating segments. Hence,
\begin{align}
\label{eq:chi}
	\sum_{i=1}^{A(\bfx)}s_i = n + \zeta(\bfx),
\end{align}
where $\zeta(\bfx)$ denotes the number of entries in $\bfx$ which are contained in exactly two alternating segments.
Define, for each $1 \leq i \leq n$
\begin{align}
	 \zeta_i(\bfx) \triangleq \begin{cases}\label{eq:chii}
	1 & x_i \text{ is contained in two maximal alternating segments} \\
	0 & \text{otherwise}
	\end{cases}
\end{align}
Thus,
\begin{small}
$$\mathop{{}\mathbb{E}}_{\bfx\in\mathbb{Z}_q^n}\left[\sum_{i=1}^{A(\bfx)}s_i\right]
= n + \mathop{{}\mathbb{E}}_{\bfx\in\mathbb{Z}_q^n}\left[\zeta(\bfx)\right]
= n + \frac{1}{q^n}\sum_{\bfx\in\mathbb{Z}_q^n}\zeta(\bfx) = n + \frac{1}{q^n}\sum_{\bfx\in\mathbb{Z}_q^n}\sum_{i=1}^n\zeta_i(\bfx)
= n + \frac{1}{q^n}\sum_{i=1}^n\sum_{x\in\mathbb{Z}_q^n}\zeta_i(\bfx).$$
\end{small}\\
Clearly, if $i\in\{1,n\}$ then $\zeta_i(\bfx)=0$ for all $\bfx\in\mathbb{Z}_q^n$.
Otherwise, $\zeta_i(\bfx)=1$ if and only if $x_{i-1},x_i$ and $x_{i+1}$ are all different.
Therefore, for $2\le i\le n-1$, there are $\binom{q}{3}\cdot 3!$ distinct ways  to select values for $x_{i-1},x_i$, and $x_{i+1}$ and $q^{n-3}$ distinct ways to select values for the other entries of $\bfx$. That is,
$$\mathop{{}\mathbb{E}}_{\bfx\in\mathbb{Z}_q^n}\left[\sum_{i=1}^{A(\bfx)}s_i\right]
= n + \frac{1}{q^n}\sum_{i=1}^n\sum_{\bfx\in\mathbb{Z}_q^n}\zeta_i(\bfx)
= n + \frac{1}{q^n}\sum_{i=2}^{n-1}\binom{q}{3}3!q^{n-3} = n + (n-2)\cdot \frac{(q-1)(q-2)}{q^2}.$$
\end{proof}

\begin{corollary}
\label{cor: si sum for q=2}
For $q=2$, we have that
$$\mathop{{}\mathbb{E}}_{\bfx\in\mathbb{Z}_2^n}\left[\sum_{i=1}^{A(\bfx)}s_i\right] = n .$$
\end{corollary}

\begin{definition}
For a sequence $\bfx=(x_1,\ldots,x_n)\in\mathbb{Z}_q^n$,  denote by ${\bfx'\in\mathbb{Z}_q^{n-1}}$ the difference vector of $\bfx$, which is defined by
$$ \bfx' \triangleq (x_2-x_1,x_3-x_2,\ldots,x_n-x_{n-1}).$$
\end{definition}

\begin{claim}
\label{cla:kAndSumSi}
For integers $n$ and $q>1$ and a sequence $\bfx\in\mathbb{Z}_q^n$,
$$\sum_{i=1}^{A(\bfx)}s_i = n + A(\bfx)-1-\mathsf{Zeros}(\bfx'),$$
where $\mathsf{Zeros}(\bfy)$ denotes the number of zeros in $\bfy$.
\end{claim}

\begin{proof}
By (\ref{eq:chi}) we have that
$$\sum_{i=1}^{A(\bfx)}s_i = n + \zeta(\bfx).$$
Since there are $A(\bfx)$ alternating segments, it follows that there are $A(\bfx)$ entries that start with a maximal alternating segment. Denote this set of entries by $\mathsf{Ind}(\bfx)$ and let $\mathsf{Ind}_1(\bfx)\subseteq\mathsf{Ind}(\bfx)$ be the set of entries $i\in\mathsf{Ind}(\bfx)$ that are contained in exactly one maximal alternating segment.
This implies that
$$\sum_{i=1}^{A(\bfx)}s_i = n + |\mathsf{Ind}(\bfx)| - |\mathsf{Ind}_1(\bfx)|.$$
Clearly, $1\in\mathsf{Ind}_1(\bfx)$. For any other index $i\in\mathsf{Ind}(\bfx)$, $x_i$ is contained in exactly one maximal alternating segment if and only if $x_i=x_{i-1}$, i.e., $x'_{i-1}=0$.
Thus,
$$\sum_{i=1}^{A(\bfx)}s_i = n + A(\bfx) - 1 - \mathsf{Zeros}(\bfx').$$
\end{proof}

\begin{claim}
\label{cla:DiffZeros}
Given two integers $n$ and $q>1$, we have that
$$\mathop{{}\mathbb{E}}_{\bfx\in\mathbb{Z}_q^n}\left[\mathsf{Zeros}(\bfx')\right] = \frac{n-1}{q}.$$

\end{claim}
\begin{proof}
By the definition of the difference vector, given $\bfy\in\mathbb{Z}_q^{n-1}$, the sequence $\bfx\in\Sigma_q^n$ such that $\bfx'=\bfy$ is defined uniquely by the selection of the first entry of $\bfx$ from $\Z_q$.  Hence, we have that for each $\bfy\in\mathbb{Z}_q^{n-1}$ there are exactly $q$ sequences $\bfx\in\mathbb{Z}_q^n$ such that $\bfx'=\bfy$. In other words, the function $f(\bfx)=\bfx'$ is a $q$ to $1$ function.
Define,
$$\mathsf{zero}_i(\bfy)\triangleq \begin{cases}
1 & y_i=0\\
0 & \text{otherwise}.
\end{cases}$$
It follows that,
\begin{align*}
   	 \mathop{{}\mathbb{E}}_{\bfx\in\mathbb{Z}_q^n}\left[\mathsf{Zeros}(\bfx')\right]
	 & = \mathop{{}\mathbb{E}}_{\bfy\in\mathbb{Z}_q^{n-1}} \left[\mathsf{Zeros}(\bfy)\right]
	 =  \frac{1}{q^{n-1}} \sum_{\bfy\in\mathbb{Z}_q^{n-1}}\mathsf{Zeros}(\bfy)
	 = \frac{1}{q^{n-1}} \sum_{\bfy\in\mathbb{Z}_q^{n-1}}\sum_{i=1}^{n-1}\mathsf{zero}_i(\bfy)\\
   	& = \frac{1}{q^{n-1}}\sum_{i=1}^{n-1}\sum_{\bfy\in\mathbb{Z}_{q}^{n-1}}\mathsf{zero}_i(\bfy).
\end{align*}
For  each $i$, the set $\{\bfy\in\Z_q^{n-1}: y_i=0\}$ is of size $\frac{q^{n-1}}{q}=q^{n-2}$. Thus,
\begin{align*}
\mathop{{}\mathbb{E}}_{\bfx\in\mathbb{Z}_q^n}\left[\mathsf{Zeros}(\bfx')\right] = \frac{1}{q^{n-1}}\sum_{i=1}^{n-1}\sum_{\bfy\in\mathbb{Z}_{q}^{n-1}}\mathsf{zero}_i(\bfy)= \frac{1}{q^{n-1}}\cdot\sum_{i=1}^{n-1} q^{n-2} = \frac{n-1}{q}.
\end{align*}
\end{proof}
By combining the results from Lemma~\ref{lam: avg sum si} and Claims~\ref{cla:kAndSumSi} and~\ref{cla:DiffZeros} we infer the following result.
\begin{corollary}
\label{cor: E[k(x)]}
For two integers $n$ and $q>1$, the average number of alternating segments of a sequence $\bfx\in\mathbb{Z}_q^n$ is
\begin{align*}
       \mathop{{}\mathbb{E}}_{\bfx\in\mathbb{Z}_q^n}\left[A(\bfx)\right] =  1 +  \frac{(n-2)(q-1)(q-2)}{q^2} + \frac{n-1}{q},
\end{align*}
and in particular for $q=2$
$$\mathop{{}\mathbb{E}}_{\bfx\in\mathbb{Z}_2^n}\left[A(\bfx)\right] = \frac{n+1}{2}.$$
\end{corollary}
\begin{proof}
For each $q>1$ we have that
\begin{align*}
    	\mathop{{}\mathbb{E}}_{\bfx\in\mathbb{Z}_q^n}\left[A(\bfx)\right]
	& = \mathop{{}\mathbb{E}}_{\bfx\in\mathbb{Z}_q^n}\left[\sum_{i=1}^{A(\bfx)}s_i\right]
	+ \mathop{{}\mathbb{E}}_{\bfx\in\mathbb{Z}_q^n}\left[\mathsf{Zeros}(\bfx')\right] - n + 1 & \text{by Claim~\ref{cla:kAndSumSi}} \\
    	& = n + \frac{(n-2)(q-1)(q-2)}{q^2} + \frac{n-1}{q} - n + 1 & \text{by Lemma~\ref{lam: avg sum si} and Claim~\ref{cla:DiffZeros}}\\
    	& = 1 +  \frac{(n-2)(q-1)(q-2)}{q^2} + \frac{n-1}{q}.&
\end{align*}
When $q=2$ the latter implies that
$$\mathop{{}\mathbb{E}}_{\bfx\in\mathbb{Z}_2^n}\left[A(\bfx)\right] = \frac{n+1}{2}.$$
\end{proof}

\begin{lemma}
\label{lem: avg num of runs}
For any two integers $n$ and $q>1$,  the average number of runs in a sequence $\bfx\in\mathbb{Z}_q^n$ is
$$\mathop{{}\mathbb{E}}_{\bfx\in\mathbb{Z}_q^n}\left[\rho(\bfx)\right] = n - \frac{n-1}{q}.$$
\end{lemma}

\begin{proof}
For a sequence $\bfx\in\mathbb{Z}_q^n$,
the number of runs in $\bfx$ is equal to the number of entries which begin a run in $\bfx$.
Clearly, $x_1$ is the beginning of the first run and by the definition of the difference vector, we have that for each $i$, $2\le i\le n$,  $x_i$ starts a run if and only if $x_{i-1}'\ne 0$.
Thus,
$$\rho(\bfx) = n - \mathsf{Zeros}(\bfx'),$$
and, by Claim~\ref{cla:DiffZeros},
$$\mathop{{}\mathbb{E}}_{\bfx\in\mathbb{Z}_q^n}\left[\rho(\bfx)\right] =  n - \mathop{{}\mathbb{E}}_{\bfx\in\mathbb{Z}_q^n}\left[\mathsf{Zeros}(\bfx')\right] = n - \frac{n-1}{q}.$$
\end{proof}

Our current goal is to evaluate $\mathop{{}\mathbb{E}}_{\bfx\in\mathbb{Z}_q^n}\left[\sum_{i=1}^{A(\bfx)}s_i^2\right]$.
Denote by $\chi(s)$ the number of maximal alternating segments of length $s$ over all the sequences $\bfx\in\mathbb{Z}_q^n$, i.e.,
$$\chi(s) = \sum_{\bfx\in\mathbb{Z}_q^n}\left|\left\{1\le i\le A(\bfx)\ : \ s_i=s\right\}\right|.$$
 It holds that
$$\mathop{{}\mathbb{E}}_{\bfx\in\mathbb{Z}_q^n}\left[\sum_{i=1}^{A(\bfx)}s_i^2\right] = \frac{1}{q^n}\sum_{\bfx\in\mathbb{Z}_2^n} \sum_{i=1}^{A(\bfx)}s_i^2 =
 \frac{1}{q^n} \sum_{s=1}^n s^2 \chi(s),$$
and the values of $\chi(s)$ for $1\le s\le n$ are given in the following lemmas.

\begin{lemma}~\label{lem: chi(1)}
If $n$ and $q>1$ are two positive integers then
$$ \chi(1) = 2q^{n-1}+(n-2)q^{n-2}.$$
\end{lemma}

\begin{proof}
Let us count the number of maximal alternating segments of length one over all the sequences $\bfx\in\mathbb{Z}_q^n$.
Consider the following two cases:  \\
    \textbf{Case $\bf 1$ - } If the alternating segment is at $x_1$, we can choose the symbols of $x_1$ in $q$ different ways.
    Since the alternating segment's length is one, i.e., $x_1=x_2$, it follows that the value of $x_2$ is determined.
    The symbols at $x_3,\ldots,x_n$ can be selected in $q^{n-2}$ different ways.
    Therefore, there are $q^{n-1}$ distinct sequences with such an alternating segment. The same arguments hold for an alternating segment at $x_n$.\\
    \textbf{Case $\bf 2$ - } If the alternating segment is at index $i, 2\le  i\le  n-1$, it must be that $x_{i-1}=x_i=x_{i+1}$. The symbol at $x_i$ can be selected in $q$ different ways and the symbols of $x_{i-1}, x_{i+1}$ are fixed. In addition.
    we can set the symbols of $\bfx$ at indices $j\notin \{i-1, i,i+1\}$ in $q^{n-3}$ different ways.
    Therefore, there are $q^{n-2}$ distinct sequences with such an alternating segment.

Thus,
$$\chi(1) = 2q^{n-1} + (n-2)q^{n-2}.$$
\end{proof}

\begin{lemma}~\label{lem: chi(n)}
For any two integers $n$ and $q>1$,
$$ \chi(n) = q(q-1).$$
\end{lemma}

\begin{proof}
Any alternating segment of length $n$ is defined by the first two symbols which must be distinct (the rest of the symbols are determined by the first two symbols). There are $q(q-1)$ different ways to select the first two symbols and hence the claim follows.
\end{proof}

For $2\le s\le n-1$ we need to consider whether the alternating segment overlaps with the preceding or the succeeding segment, or not. To this end, we distinguish between the maximal alternating segments of length $s$ as follows
\begin{enumerate}
\item[] $\chi_1(s) $ - The number of alternating segments that do not overlap with the preceding segment and the succeeding segments.
\item[] $\chi_2(s) $ - The number of alternating segments that overlap with the preceding segment and the succeeding segments.
\item[] $\chi_3(s) $ - The number of alternating segments that overlap only with the succeeding segment.
\item[] $\chi_4(s) $ - The number of alternating segments that overlap only with the preceding segment.
\end{enumerate}
\begin{claim}~\label{cla: alt chi(s) calculation}
If $n, q>1$ are integers and $2\le s\le n-1$ then,
\begin{enumerate}
\item $\chi_1(s)  = 2(q-1)q^{n-s} + (n-s-1)(q-1)q^{n-s-1}.$
\item $\chi_2(s)  = (n-s-1)(q-1)(q-2)^2q^{n-s-1}.$
\item $\chi_3(s)  =  (q-1)(q-2)q^{n-s} + (q-1)(q-2)(n-s-1)q^{n-s-1}.$
\item $\chi_4(s)  =  (q-1)(q-2)q^{n-s} + (q-1)(q-2)(n-s-1)q^{n-s-1}.$
\end{enumerate}
\end{claim}

\begin{proof}
\begin{enumerate}
\item To count the number of maximal alternating segments of length $s$ that do not overlap with the preceding segment and the succeeding segment we distinguish two distinct cases.\\
\textbf{Case $\bf 1$ - } If the alternating segment is at the beginning of the sequence, then there are $q(q-1)$ distinct ways to select the symbols of the segment. The symbol after the segment is determined (and is equal to the last symbol of the discussed alternating segment) in order to prevent an overlap and the other symbols can be chosen in $q^{n-s-1}$ different ways. Hence, the number of different sequences with such segments is $(q-1)q^{n-s}$. The same arguments hold for an alternating segment at the end of the sequence.\\
\textbf{Case $\bf 2$ - } If the alternating segment is not at the edges of the sequence, then there are $n-s-1$ possible positions to start the alternating segment, and  $q(q-1)$ ways to choose the two symbols of the alternating segment. The symbol preceding and the symbol succeeding the alternating segment are determined.  The other symbols can be chosen in $q^{n-s-2}$ distinct ways and hence the number of different alternating segments is $(n-s-1)(q-1)q^{n-s-1}$.

Thus,
$$\chi_1(s)  = 2(q-1)q^{n-s} + (n-s-1)(q-1)q^{n-s-1}.$$
\item A maximal alternating segment that overlaps with the preceding segment and the succeeding segment can not be at the sequence edges. Hence, there are $n-s-1$ possible positions to start the alternating segment and the symbols of the segment can be chosen in $q(q-1)$ different ways. In order to overlap with the preceding (succeeding, respectively)  segment, the symbol before (after, respectively) the segment must be different from the two symbols of the segment. Therefore, there are $(q-2)^2$ options to choose the symbol before and the symbol after the segment.  In addition, the rest of the sequence can be chosen in $q^{n-s-2}$ different ways and hence
$$\chi_2(s)  = (n-s-1)(q-1)(q-2)^2q^{n-s-1}.$$
\item Since the alternating segment must intersect with the succeeding segment, it can not be the last alternating segment, that is, the segment ends at index $j < n$.
To count the number of maximal alternating segments of length $s$ that overlap only with the succeeding segment we consider two distinct cases.\\
\textbf{Case $\bf 1$ - } If the alternating segment is at the beginning of the sequence then there are $q(q-1)$ different ways to choose the symbols for it and the symbol after the segment must be different from the two symbols of the alternating segment so there are $(q-2)$ options to select it. The other symbols can be chosen in $q^{n-s-1}$ different ways. Hence, the number of different segments is $(q-1)(q-2)q^{n-s}$.\\
\textbf{Case $\bf 2$ - }  If the alternating segment does not start  at the beginning of the sequence, since the segment ends at index $j<n$, it follows that there are $(n-s-1)$ possible locations to start the segment. There are  $q(q-1)$ different ways to select the symbols for the alternating segment. The symbol before the alternating segment is determined in order to prevent an overlap with the previous segment and the symbol after the segment must be different from the two symbols of the alternating segment and hence  there are $(q-2)$ ways to choose it. The other symbols can be chosen in $q^{n-s-2}$ different ways and hence the number of different segments is $q^{n-s-1}(q-1)(q-2)(n-s-1)$.\\
Thus,
$$\chi_3(s) = (q-1)(q-2)q^{n-s} + (q-1)(q-2)(n-s-1)q^{n-s-1}.$$
\item Clearly, the number of maximal alternating segments of length $s$ that overlap only with the succeeding segment is equal to the number alternating segments of length $s$ that overlap only with the preceding segment.
 \end{enumerate}
\end{proof}

\begin{lemma}~\label{lem: chi(s)}
In $n,q>1$ are integers and $2\le s \le n-1$ then
$$\chi(s) = 2(q-1)^2q^{n-s} + (n-s-1)(q-1)^3q^{n-s-1}.$$
\end{lemma}
\begin{proof}
By Claim~\ref{cla: alt chi(s) calculation},
\begin{align*}
\chi(s) & = \chi_1(s) + \chi_2(s) + \chi_3(s) + \chi_4(s)\\
 & = 2(q-1)q^{n-s} + (n-s-1)(q-1)q^{n-s-1} + (n-s-1)(q-1)(q-2)^2q^{n-s-1} \\
 & + 2(q-1)(q-2)q^{n-s} + 2(n-s-1)(q-1)(q-2)q^{n-s-1}\\
 & = 2(q-1)^2q^{n-s} + (n-s-1)(q-1)q^{n-s-1}\left(1+(q-2)^2 + 2(q-2)\right) \\
 & =  2(q-1)^2q^{n-s} + (n-s-1)(q-1)q^{n-s-1}\left(q^2-2q+1)\right) \\
& =  2(q-1)^2q^{n-s} + (n-s-1)(q-1)^3q^{n-s-1}.\
\end{align*}
\end{proof}

\begin{lemma}
\label{lem: avg sum of si^2}
If $n, q>1$ are integers then,
$$\mathop{{}\mathbb{E}}_{\bfx\in\mathbb{Z}_q^n}\left[ \sum_{i=1}^{A(\bfx)}s_i^2\right]
=   \frac{n(4q^2-3q+2)}{q^2}+
   \frac{6 q - 4}{q^2}
   - 4
   -\frac{2}{q - 1} \left( 1 - \frac{1}{q^n}\right).$$
\end{lemma}

\begin{proof}
We have that
\begin{align*}
\mathop{{}\mathbb{E}}_{\bfx\in\mathbb{Z}_q^n}\left[\sum_{i=1}^{A(\bfx)}s_i^2\right] & = \frac{1}{q^n}\sum_{\bfx\in\mathbb{Z}_q^n} \sum_{i=1}^{A(\bfx)}s_i^2 =
 \frac{1}{q^n} \sum_{s=1}^n s^2 \chi(s) = \frac{\chi(1)}{q^n} + \frac{n^2\chi(n)}{q^n} + \frac{1}{q^n} \sum_{s=2}^{n-1} s^2 \chi(s).
 \end{align*}
Let us first calculate $\sum_{s=2}^{n-1} s^2 \chi(s)$. By Lemma~\ref{lem: chi(s)},
\begin{align*}
 \sum_{s=2}^{n-1} s^2 \chi(s) & = \sum_{s=2}^{n-1}s^2\left( 2(q-1)^2q^{n-s} + (n-s-1)(q-1)^3q^{n-s-1}\right) \\
 & = 2(q-1)^2\sum_{s=2}^{n-1}s^2q^{n-s} + (q-1)^3\sum_{s=2}^{n-1}(n-s-1)s^2q^{n-s-1}.
 \end{align*}
It can be verified that
 \begin{align*}
\sum_{s=2}^{n-1} s^2 \chi(s) =  \frac{2q^3-q^3n^2(q-1)^2+ q^n(2-2q(3+q(2q-3))+n(q-1)(1+q(4q-3)))}{(q-1)q^{2}}
 \end{align*}
 and after rearranging the latter, we obtain that
 \begin{align*}
  \sum_{s=2}^{n-1} s^2 \chi(s) =  nq^{n-2} (4q^2-3q+1)
   -n^2q(q-1)
    -2q^{n-2}\cdot \frac{(2q-1)(q^2-q+1)}{(q-1)} + \frac{2}{q-1}.
\end{align*}
Hence,
\begin{align*}
\mathop{{}\mathbb{E}}_{\bfx\in\mathbb{Z}_q^n}\left[\sum_{i=1}^{A(\bfx)}s_i^2\right] & = \frac{\chi(1)}{q^n} + \frac{n^2\chi(n)}{q^n} + \frac{1}{q^n} \sum_{s=2}^{n-1} s^2 \chi(s)\\
& =  \frac{2q^{n-1}+(n-2)q^{n-2}}{q^n} + \frac{n^2q(q-1)}{q^n}+ \frac{nq^{n-2} (4q^2-3q+1)}{q^n} \\
  & -\frac{n^2q(q-1)}{q^n}
    -2q^{n-2}\cdot \frac{(2q-1)(q^2-q+1)}{q^n(q-1)} + \frac{2}{q^n(q-1)} \\
    & =  \frac{n(4q^2-3q+2)}{q^2}+
     \frac{2}{q}
     -\frac{2}{q^2}\\
    &- \frac{2(2q-1)(q^2-q+1)}{q^2(q-1)} +
    \frac{2}{q^n(q-1)} \\
      & =  \frac{n(4q^2-3q+2)}{q^2}+
   \frac{6 q - 4}{q^2}
   - 4
   -\frac{2}{q - 1} \left( 1 - \frac{1}{q^n}\right).
    \end{align*}
\end{proof}

\begin{theorem}
\label{the: avg l-ball}
If $n,q>1$ are integers, then
$$\mathop{{}\mathbb{E}}_{\bfx\in\mathbb{Z}_q^n}\left[\left|\cL_1(\bfx)\right|\right]
= n^2\left(q+\frac{1}{q} -2\right) - \frac{n}{q} - \frac{(q-1)(q-2)}{q^2}
	 +3 - \frac{3}{q} + \frac{2}{q^2} + \frac{q^n-1}{q^n(q-1)}.$$
\end{theorem}

\begin{proof}
By (\ref{eq: avg L1}) we have that
\begin{small}
\begin{align*}
    	\mathop{{}\mathbb{E}}_{\bfx\in\mathbb{Z}_q^n}\left[|\cL_1(\bfx)|\right]
	&  = \left(nq-n-1\right)\mathop{{}\mathbb{E}}_{\bfx\in\mathbb{Z}_q^n}\left[\rho(\bfx)\right]
	+ 2 - \frac{1}{2}\mathop{{}\mathbb{E}}_{\bfx\in\mathbb{Z}_q^n}\left[\sum_{i=1}^{A(\bfx)}s_i^2\right]
	+ \frac{3}{2}\mathop{{}\mathbb{E}}_{\bfx\in\mathbb{Z}_q^n}\left[\sum_{i=1}^{A(\bfx)}s_i\right]
	- \mathop{{}\mathbb{E}}_{\bfx\in\mathbb{Z}_q^n}\left[A(\bfx)\right].
\end{align*}
\end{small}
Using Corollary~\ref{cor: E[k(x)]} and Lemmas~\ref{lam: avg sum si},~\ref{lem: avg num of runs}, and~\ref{lem: avg sum of si^2} we infer that
\begin{align*}
    	\mathop{{}\mathbb{E}}_{\bfx\in\mathbb{Z}_q^n}\left[\left|\cL_1(\bfx)\right|\right] &  = \left(nq-n-1\right)\left(n-\frac{n-1}{q}\right) + 2 \\
    	& -\frac{1}{2}\left(\frac{n(4q^2-3q+2)}{q^2}+
   \frac{6 q - 4}{q^2}
   - 4
   -\frac{2}{q - 1} \left( 1 - \frac{1}{q^n}\right)\right)\\
    	& +\frac{3}{2}\left( n + (n-2)\cdot \frac{(q-1)(q-2)}{q^2}\right)
	 - 1 -  \frac{(n-2)(q-1)(q-2)}{q^2} - \frac{n-1}{q}\\
	 &= n^2\left(q+\frac{1}{q} -2\right) - \frac{n}{q} - \frac{(q-1)(q-2)}{q^2}
	 +3 - \frac{3}{q} + \frac{2}{q^2} + \frac{q^n-1}{q^n(q-1)}.
\end{align*}
\end{proof}

\section{Binary Anticodes with Diameter one}
\label{sec:anticode_size}

Before presenting the analysis of the anticodes under the FLL metric, we state the following lemma, which was proven in~\cite{L01} and will be used in some of the proofs in this section.
\begin{lemma}~\label{lem: del/ins intersection}
If $\bfx,\bfy \in \Z_2^n$ are distinct words, then
$$|\cD_1(\bfx)\cap \cD_1(\bfy)|\le 2\ \text{ and }\  |\cI_1(\bfx)\cap \cI_1(\bfy)|\le 2.$$
\end{lemma}

\begin{definition}
An \emph{anticode of diameter $t$} in $\Z_q^n$ is a subset $\cA\subseteq \Z_q^n$
such that for any $\bfx,\bfx'\in \cA$, $d_\ell(\bfx,\bfx')\le t$.
We say that $\cA$ is a \emph{maximal anticode} if
there is no other anticode of diameter $t$ in  $\Z_q^n$ which contains $\cA$.
\end{definition}
Next, we present tight lower and upper bounds on the size of maximal binary anticodes of diameter one in the FLL metric. To prove these bounds we need some useful properties of anticodes with diameter one in the FLL metric.

\begin{lemma}
\label{lem: suffix 00}
If an anticode $\cA$ of diameter one contains three distinct words with the suffix 00 then there is at most one word in $\cA$ with the suffix 01.
\end{lemma}
\begin{proof}
Let $\bfa,\bfa',\bfa''\in\cA$ be three words with the suffix {00} and assume to the contrary that there exist two distinct words $\bfb,\bfb'\in \cA$ with the suffix {01}. Let $\bfy \in \mathcal{LCS}(\bfa,\bfb)$; by Corollary~\ref{cor: LCS length} the length of $\bfy$ is $n-1$ and since $\bfa$ ends with {00}, $\bfy$ must end with {0} which implies that $\bfy=\bfb_{[1,n-1]}$. By the same arguments $\bfy\in \mathcal{LCS}(\bfb,\bfa')$ and $\bfy\in \mathcal{LCS}(\bfb, \bfa'')$.
Similarly,
$$\bfy' = \bfb'_{[1,n-1]}\in \mathcal{LCS}(\bfb',\bfa,\bfa',\bfa'').$$
Hence, $\bfa,\bfa',\bfa''\in \cI_1(\bfy) \cap \cI_1(\bfy') $ which is a contradiction to Lemma~\ref{lem: del/ins intersection}. Thus, $\cA$ contains at most one word with the suffix {01}.
\end{proof}

\begin{lemma}~\label{lem: suffix 01}
If an anticode $\cA$ of diameter one contains three distinct words with the suffix {01}, then there is at most one word in $\cA$ with the suffix {00}.
\end{lemma}
\begin{proof}
Let $\bfa,\bfa',\bfa''\in\cA$ be three words with the suffix {01} and assume to the contrary that there exist two distinct words $\bfb,\bfb'\in \cA$ with the suffix {00}.
For $\bfy\in \mathcal{LCS}(\bfa,\bfb)$, by Corollary~\ref{cor: LCS length} the length of $\bfy$ is $n-1$ and since $\bfb$ ends with {00}, $\bfy$ must end with {0} which implies that $\bfy=\bfa_{[1,n-1]}$.
By the same arguments $\bfy\in \mathcal{LCS}(\bfa,\bfb')$.
Similarly,
\begin{align*}
\bfy' = \bfa'_{[1,n-1]} &\in \mathcal{LCS}(\bfa',\bfb,\bfb')\\
\bfy'' = \bfa''_{[1,n-1]} &\in \mathcal{LCS}(\bfa'',\bfb,\bfb').
\end{align*}
 Hence, $\bfy, \bfy', \bfy''\in \cD_1(\bfb) \cap \cD_1(\bfb') $ which is a contradiction to Lemma~\ref{lem: del/ins intersection}. Thus, $\cA$ contains at most one word with the suffix {00}.
\end{proof}

\begin{lemma}~\label{lem: suffixes}
Let $\cA$ be an anticode of diameter one. If ${\bfa,\bfa'\in \cA}$ are two distinct words that end with {00} and $\bfb,\bfb'\in\ \cA$ are two distinct words that end with {01},
then $\bfa_{[1,n-1]}\ne \bfb_{[1,n-1]}$ or  $\bfa'_{[1,n-1]}\ne\bfb'_{[1,n-1]}$.
\end{lemma}
\begin{proof}
Assume to the contrary that there exist $\bfa,\bfa',\bfb,\bfb'\in \cA$ such that  ${\bfa_{[1,n-1]}=\bfb_{[1,n-1]}=\bfy}0$ and  $\bfa'_{[1,n-1]}=\bfb'_{[1,n-1]}=\bfy'0$,
$\bfa,\bfa'$ end with {00} and $\bfb,\bfb'$ end with {01}. Let,
\begin{align*}
\bfa \  &= a_1\ a_2 \ldots a_{n-2}\  {0\  0}   = \bfy\ \   {0\ 0} \\
\bfa'&= a'_1\ a'_2 \ldots a'_{n-2}\  {0\  0}=\bfy'\  {0\ 0} \\
\bfb \ &= a_1\ a_2 \ldots a_{n-2}\  {0\  1}  =\bfy\ \ {0\ 1} \\
\bfb'&= a'_1\ a'_2 \ldots a'_{n-2}\  {0\ 1}=\bfy'\  {0\ 1}.
\end{align*}
Notice that since the FLL distance between any two words in $\cA$ is one, it follows that
the Hamming weight of any two words can differ by at most one, which implies that $\text{wt}(\bfy)=\text{wt}(\bfy')$
(by considering the pairs $\bfa,\bfb'$ and $\bfa',\bfb$).
Clearly, $\bfy{0}\in \mathcal{LCS}(\bfa',\bfb)$ which implies that $\bfa'$ can be obtained from~$\bfb$ by deleting the last~{1} of $\bfb$ and then inserting {0} into the LCS. Hence, there exists an index $0\le j\le n-2$ such that

\begin{equation}
\label{eq: a' from a}
a_1 a_2 \ldots a_j{0} a_{j+1} \ldots a_{n-2}{0}  = a_1'a_2'\ldots a'_j  a'_{j+1} \ldots a'_{n-2}{00}.
\end{equation}
Similarly, $\bfa$ can be obtained from $\bfb'$, i.e., there exists an index $0\le i\le n-2$ such that

\begin{equation}~\label{eq: a from a'}
a_1'a_2'\ldots a'_i {0} a'_{i+1} \ldots a'_{n-2}{0}  = a_1 a_2 \ldots a_i a_{i+1} \ldots a_{n-2} {00}.
\end{equation}
Assume w.l.o.g. that $i\le j$. (\ref{eq: a' from a}) implies that $a_{r}=a_{r'}$ for $1\le r\le j$. In addition, $a_{n-2}={0}$ by~(\ref{eq: a' from a}) and $a_{n-2}'={0}$ by (\ref{eq: a from a'}). By assigning  $a_{n-2} = a_{n-2}' = {0}$ into (\ref{eq: a' from a}) and (\ref{eq: a from a'}) we obtain that $a_{n-3}=a_{n-3}'={0}$.  Repeating this process implies that $a_{r}=a_{r'}={0}$ for $j+1 \le r\le n-2$.
Thus, we have that $\bfy=\bfy'$ which is a contradiction.
\end{proof}

\begin{definition}
For an anticode ${\cA}\subseteq\mathbb{Z}_2^n$, the  \emph{puncturing of $\cA$ in the $n$-th coordinate}, $\cA'$, is defined by
$$\cA' \triangleq \left\{ \bfa_{[1:n-1]}\ : \ {\bfa}\in\cA \right\}.$$
\end{definition}

\begin{lemma}~\label{lem: anticode with fix last symbol}
Let $\cA\subseteq\mathbb{Z}_2^n$ be an anticode of diameter one. If the last symbol in all the words in~$\cA$ is the same symbol ${\sigma \in \mathbb{Z}_2^n}$, then $\cA'$ is an anticode of diameter one and ${|\cA'|=|\cA|}$.
\end{lemma}

\begin{proof}
Let $\bfa,\bfb\in\cA$ be two different words and let ${\bfy \in \mathcal{LCS}(\bfa_{[1:n-1]}, \bfb_{[1:n-1]})}$. By (\ref{eq: recursive LCS}), $\mathsf{LCS}(\bfa,\bfb)\le |\bfy|+1$ and since $d_\ell(\bfa,\bfb)=1$, Corollary~\ref{cor: LCS length} implies that ${|\bfy|\ge n-2}$ and that $$d_\ell(\bfa_{[1:n-1]},\bfb_{[1:n-1]})\le 1.$$ Hence, $\cA$ is an anticode of diameter one.
Since any two distinct words $\bfa, \bfb  \in \cA$ end with the symbol $\sigma$, it follows that $\bfa_{[1:n-1]}\ne \bfb_{[1:n-1]}$ and thus $|\cA|=|\cA'|$.
\end{proof}

\begin{lemma}~\label{lem: anticode with alt suffix}
Let $\cA$ be an anticode of diameter one. If the suffix of each word in $\cA$ is either  {01} or  {10}, then $\cA'$ is an anticode of diameter one and $|\cA'|=|\cA|$.
\end{lemma}

\begin{proof}
Let $\bfa,\bfb\in\cA$ be two different words and let $\bfy \in \mathcal{LCS}(\bfa_{[1:n-1]}, \bfb_{[1:n-1]})$.  By (\ref{eq: recursive LCS}), $\mathsf{LCS}(\bfa,\bfb)\le |\bfy|+1$ and since $d_\ell(\bfa,\bfb)=1$, it follows that $|\bfy|\ge n-2$ and that $$d_\ell(\bfa_{[1:n-1]},\bfb_{[1:n-1]})\le 1.$$ Hence, $\cA'$ is an anticode of diameter one.
If $\bfa$ and $\bfb$ end with the same symbol $\sigma\in\{{0,1}\}$, then $\bfa_{[1:n-1]}\ne \bfb_{[1:n-1]}$. Otherwise, one of the words has the suffix {01} and the other has the suffix {10}. That is, $a_{n-1}\ne b_{n-1}$ and therefore ${\bfa_{[1:n-1]}\ne \bfb_{[1:n-1]}}$ and thus, $|\cA'|=|\cA|$.
\end{proof}

\subsection{Upper Bound}
\label{sec:upper_anticodes}

\begin{theorem}
Let $n>1$ be an integer and let ${\cA\subseteq\Z_2^n}$ be a maximal anticode of diameter one.
Then, $|\cA|\le n+1$, and there exists a maximal anticode with exactly $n+1$ codewords.
\end{theorem}

\begin{proof}
Since two words $\bfx,\bfy$ such that $\bfx$ ends with {00} and $\bfy$ ends with {11} are at FLL distance at least $2$, w.l.o.g. assume that $\cA$ does not contain codewords that end with~{11}. It is easy to verify that the theorem holds for $n\in\{2,3,4\}$. Assume that the theorem does not hold and let $n^*>4$ be the smallest integer such that there exists an anticode  $\cA\subseteq\mathbb{Z}_2^{n^*}$ such that $|\cA|= n^*+2$. Since there are only three possible options for the last two symbols of codewords in $\cA$ (00, 01, or 10) and $|\cA|\ge 7$, it follows that there exist three different codewords in $\cA$ with the same suffix of two symbols.   \\
\textbf{Case $\bf 1$ - } Assume $\bfx,\bfy,\bfz\in \cA$ are  three different words with the suffix {00}. By Lemma~\ref{lem: suffix 00}, there exists at most one codeword in $\cA$ with the suffix {01} and since $\cA$ does not contain codewords with the suffix {11}, there exists at most one codeword in $\cA$ that ends with the symbol {1}. That is, there exist at least $n^*+1$ codewords with {0} as the last symbol. Denote such a set with $n^*+1$ codewords by $\cA_1$. As a subset of the anticode $\cA$, $\cA_1$ is also an anticode and hence by Lemma~\ref{lem: anticode with fix last symbol}, $\cA_1'$ is an anticode of length $n^*-1$ and size $n^*+1$ which is a contradiction to the minimality of $n^*$.\\
\textbf{Case $\bf 2$ - } Assume $\bfx,\bfy,\bfz\in \cA$ are  three different words with the suffix {01}.
By Lemma~\ref{lem: suffix 01}, there exists at most one codeword in $\cA$ with the suffix {00} and since $\cA$ does not contain codewords with the suffix {11} there exist $n^*+1$ codewords that end with either {01} or {10}.  Denote this set of $n^*+1$ codewords as $\cA_1$. As a subset of the anticode $\cA$, $\cA_1$ is also an anticode and hence by Lemma~\ref{lem: anticode with alt suffix}, $\cA_1'$ is an anticode of length $n^*-1$ and size $n^*+1$ which is a contradiction to the minimality of~$n^*$.\\
\textbf{Case $\bf 3$ - } Assume $\bfx,\bfy,\bfz\in \cA$ are  three different words with the suffix {10}.
By the previous two cases, there exist at most two codewords in $\cA$ with the suffix {00} and at most two codewords with the suffix {01}. Since there are no codewords with the suffix {11}, it follows that the number of words that end with {1} is at most two.
If there exist at most one codeword in $\cA$ that ends with {1}, then there are $n^*+1$ codewords in $\cA$ that end with {0} and as in the first case, this leads to a contradiction.
Otherwise there are exactly two codewords in $\cA$ with the suffix {01}. If there are less than two codewords with the suffix {00}, then, the number of codewords with suffixes {01} and {10} is at least $n^*+1$ and similarly to Case $2$, this is a contradiction to the minimality of $n^*$.
Hence, there exist exactly two codewords in $\cA$ with the suffix {00}.
There are exactly $n^*-2$ codewords in $\cA$ with the suffix {10} and two more codewords with the suffix {01}.
By Lemma~\ref{lem: anticode with alt suffix} the words in $\cA'$ that were obtained from these $n^*$ codewords are all different and have FLL distance one from each other. In addition, by Lemma~\ref{lem: suffixes},
 the prefix of length $n^*-1$ of at least one of the codewords that end with {00} is different from the prefixes of length $n^*-1$ of the codewords that end with {01}. This prefix also differ from the prefixes of the codewords that end with {10}.
Therefore, $\cA'$ is an anticode with $n^*+1$ different codewords which is a contradiction to the minimality of~$n^*$.

Note that the set $\cA=\left\{a\in\mathbb{Z}_2^n \ : \ \text{wt}(a)\le1\right\}$ is an  anticode of diameter one with exactly $n+1$ codewords. Thus, the maximum size of an anticode of diameter one is~${n+1}$.
\end{proof}

\subsection{Lower Bound}
\label{sec:lower_anticodes}

\begin{theorem}
Let $n>2$ be a positive integer and let $\cA \subseteq\mathbb{Z}_2^n$ be a maximal anticode of diameter one,
then $|\cA|\ge 4$ and there exists a maximal anticode with exactly 4 codewords.
\end{theorem}

\begin{proof}
For $n=3$ the maximal anticodes are
\begin{align*}
\cA_1 = \{000,\ 001, \ 010,\ 100\} &\ \ \ \cA_2 = \{001,\ 010, \ 100,\ 101\} & \cA_3 = \{001,\ 010, \ 011,\ 101\} \\
\cA_4 = \{010,\ 011, \ 101,\ 110\} &\  \ \ \cA_5 = \{011,\ 101, \ 110,\ 111\} & \cA_6 = \{010,\ 100, \ 101,\ 110\}
\end{align*}
and all of them have size $4 = n + 1$. Assume that the theorem does not hold and let $n^*>3$ be the smallest integer such that there exists a maximal anticode $\cA\subseteq \mathbb{Z}_2^{n^*}$ with less than four codewords. For each $\bfx\in\mathbb{Z}_2^{n^*}$ there exists a sequence $\bfy\in\mathbb{Z}_2^{n^*}$ such that $d_\ell(\bfx,\bfy)=1$ and hence $|\cA|>1$. If $\cA = \{\bfx,\bfy\}\subseteq \mathbb{Z}_2^{n^*}$ by the definition of an anticode $d_{\ell}(\bfx,\bfy)=1$ and $\mathsf{LCS}(\bfx,\bfy)=n-1$. For $\bfz\in\mathcal{ LCS}(\bfx,\bfy)$, by (\ref{eq: insertion ball size}), the insertion ball of radius one centered at $\bfz$ contains $n^*-1>2$ codewords in addition to $\bfx$ and $\bfy$ and each of them can be added into~$\cA$. Hence, $\cA$ is an anticode of diameter one with three codewords. We will prove that there exists a word that can be added into $\cA$ which is a contradiction to the maximality of~$\cA$. Consider the following cases:\\
\textbf{Case $\bf 1$ - } If all the codewords in $\cA$ have the same last symbol $\sigma\in\mathbb{Z}_2$, then by Lemma~\ref{lem: anticode with fix last symbol}, $\cA'\subseteq \mathbb{Z}_2^{n^*-1}$, is an anticode of diameter one that contains three codewords. Since $n^*$ is the smallest integer for which there exists a maximal anticode with less than four codewords, $\cA'$ is not maximal. That is, there exists a word $\bfx'\in\mathbb{Z}_2^{n^*-1}$ such that $\cA'\cup \{\bfx'\}$ is an anticode of diameter one. It can be readily verified that $\bfx'\sigma\notin\cA$ and that $\cA\cup\{\bfx'\sigma\}$ is an anticode of diameter one  which is a contradiction to the maximality of $\cA$.  \\
\textbf{Case $\bf 2$ - } If all the codewords in $\cA$ have the same first symbol $\sigma\in\mathbb{Z}_2$ then a contradiction is obtained by symmetrical arguments to those presented in Case $1$. \\
\textbf{Case $\bf 3$ - } Assume all the words in $\cA$  neither have the same first symbol nor the same last symbol. Let $|\cA| = \{\bfx,\bfy,\bfz\}$ and assume w.l.o.g. that $\bfx$ and $\bfy$ are codewords that end with $0$ and that~$\bfz$ ends with $1$. If $|\cA'|=3$, then $\bfz_{[1,n^*-1]}\ne \bfx_{[1,n^*-1]}$ and  $\bfz_{[1,n^*-1]}\ne \bfy_{[1,n^*-1]}$. Hence the word $\bfz_{[1,n^*-1]} 0$ is not in $\cA$ and it is easy to verify that it has distance one from each codeword in $\cA$, which is a contradiction.  Otherwise, since $\bfx_{[1,n^*-1]}\ne \bfy_{[1,n^*-1]}$, it must hold that $\bfz_{[1,n^*-1]}$ is equal either to $\bfx_{[1,n^*-1]}$ or to $\bfy_{[1,n^*-1]}$. Assume w.l.o.g. that  $\bfz_{[1,n^*-1]} = \bfx_{[1,n^*-1]}$, then $\bfx$ and $\bfz$ have the same first symbol $\sigma$ and hence $\bfy$ must begin with  $\overline{\sigma} = 1-\sigma$. The three codewords can be described as follows:
\begin{align*}
\bfx & = \sigma \textcolor{red}{x_2 x_3 \ldots x_{n^*-1}} 0\\
\bfy & =  \overline{\sigma} \textcolor{cyan}{y_2 y_3 \ldots y_{n^*-1}} 0\\
\bfz & =  \sigma \textcolor{red}{x_2 x_3 \ldots x_{n^*-1}}  1.
\end{align*}
Since $\bfy$ and $\bfz$ have different first and last symbols, their LCS must be equal to the suffix of length $n^*-1$ of one word and to the prefix of length $n^*-1$ of the other word. If
$$\bfz_{[1,n^*-1]} = \sigma  \textcolor{red}{a_2 a_3 \ldots a_{n^*-1}}\in \mathcal{LCS}(\bfy,\bfz),$$
then $\bfz_{[1,n^*-1]}$ is a common LCS of the three codewords $\bfx,\bfy$, and $\bfz$ and hence any word from $\cI_1(\bfz_{[1,n^*-1]})$ has distance one from all the words in $\cA$. Since, by (\ref{eq: insertion ball size}),
$$|\cI_1(\bfz_{[1,n^*-1]})| = n^*+1 \ge 4,$$
there is a word different from $\bfx,\bfy$ and $\bfz$ that can be added into $\cA$. In the other case,
$$\overline{\sigma} \textcolor{cyan}{y_2 y_3 \ldots y_{n^*-1}} = \bfy_{[1,n-1]} = \bfz_{[2,n]} = \textcolor{red}{x_2 x_3 \ldots x_{n^*-1}}  1\in \mathcal{LCS}(\bfy,\bfz)$$
and hence the codewords $\bfx$ and $\bfz$ can be written as
\begin{align*}
\bfx & =\sigma \overline{\sigma} \textcolor{cyan}{y_2 y_3 \ldots y_{n^*-2}} 0\\
\bfz & =  \sigma \overline{\sigma} \textcolor{cyan}{y_2 y_3 \ldots y_{n^*-2} y_{n^*-1}}
\end{align*}
and the word
$$\bfw = \sigma \overline{\sigma} \textcolor{cyan}{y_2 y_3 \ldots y_{n^*-1}} 0$$
is a common SCS of $\bfx,\bfy$, and $\bfz$.  If $\rho(\bfw)>3$  then there is a word in $\cD_1(\bfw)$ that is different from $\bfx,\bfy$, and $\bfz$ that can be added into $\cA$, which is again a contradiction. Otherwise, since the first two symbols of $\bfw$ are different and the last two symbols are also different, it holds that $\rho(\bfw)=3$. It is easy to verify that
$$\cA = \{0\underbrace{11\ldots 1}_{n^*-2 \text{ times}} 0, 0 \underbrace{11\ldots 1}_{n^*-1 \text{ times}}, \underbrace{11\ldots 1}_{n^*-1 \text{ times}} 0\}$$
 and that $\underbrace{11\ldots 1}_{n^*-2 \text{ times}}01$ can be added into $\cA$, which is a contradiction to the minimality of $\cA$. 
To see that the given bound is tight, one can simply consider the set of codewords that consist from the binary representation of length $n^*$ of the numbers $2,3,5,6$ that is, the set
$$
\cA= \{ \underbrace{0\ldots0}_{n^*-3}010,\  \underbrace{0\ldots0}_{n^*-3}011,\  \underbrace{0\ldots0}_{n^*-3}101,\ \underbrace{0\ldots0}_{n^*-3}110\}
$$
and verify that it is indeed a maximal anticode of diameter one.
\end{proof}

\section{Conclusion}
In this paper we studied the size of balls with radius one and the anticodes of diameter one under the FLL metric.
In particular we give explicit expressions for the maximum size of a ball with radius one and the minimum size
of a ball of any given radius in the FLL metric over $\Z_q$. We also found the average size of a $1$-ball in
the FLL metric. Finally, we considered the related concept of anticode in the FLL distance and we found that
the maximum and minimum size of a binary maximal anticode of diameter one are $n+1$ and $4$, respectively.
The latter can be extended to a non-binary alphabet and while the minimum size of a maximal anticode with diameter
one is $4$ for any alphabet size $q$, the maximum size of a maximal anticode with diameter one is $n(q-1)+1$.
The results in this paper were presented in part at the IEEE International Symposium on Information Theory (ISIT), 2021~\cite{BEY21}.
Recently, based on these results,  G. Wang and Q. Wang~\cite{WW22} extended the analysis of $1$-FLL balls by proving
that the size of the $1$-FLL balls is highly concentrated around its mean using Azuma’s inequality~\cite{AS16}.

\section*{Appendix}
\subsection*{Proof of Lemma~\ref{lam : max t}}
Let $\alpha>1$ be some integer and define  $\mathsf{diff} \triangleq \left|L_1\left(\bfx^{(\alpha)}\right)\right|-\left|L_1\left(\bfx^{(\alpha-1)}\right)\right|$. We will prove that $\mathsf{diff} >0$ if and only if $n>2\alpha(\alpha-1)$ by proving that $\mathsf{diff} >0$ for any $n>2\alpha(\alpha-1)$ and that $\mathsf{diff} <0$ for any $\alpha<n<2\alpha(\alpha-1)$. Before we analyze each case we present different expression for $ \left|L_1\left(\bfx^{(\alpha)}\right)\right|$ that will be in use within the proof. By Lemma~\ref{lem: k balanced ball size}, if $n$ is divisible by $\alpha$, then
\begin{align*}
\left|\cL_1\left(\bfx^{(\alpha)}\right)\right|
 & =  (n+1-\alpha)(n-1)+2 - \frac{\alpha}{2}\left(\frac{n}{\alpha}-1\right)
\left(\frac{n}{\alpha}-2\right)\\
& =  (n+1-\alpha)(n-1)+2 -\frac{n^2}{2\alpha}+\frac{3n}{2}-\alpha.
\end{align*}
Otherwise when $n$ is not divisible by $\alpha$, we have that $\left\lceil\frac{n}{\alpha}\right\rceil = \frac{n-k_{\alpha}}{\alpha}+1$ and hence, by Lemma~\ref{lem: k balanced ball size},

\begin{small}
\begin{align*}
\left|\cL_1\left(\bfx^{(\alpha)}\right)\right| & = (n+1-\alpha)(n-1) +2
 - \frac{k_\alpha}{2}\left(\left\lceil\frac{n}{\alpha}\right\rceil-1\right)
\left(\left\lceil\frac{n}{\alpha}\right\rceil-2\right)
 - \frac{\alpha-k_\alpha}{2}\left(\left\lceil\frac{n}{\alpha}\right\rceil-2\right)
 \left(\left\lceil\frac{n}{\alpha}\right\rceil-3\right)\\
 & = (n+1-\alpha)(n-1)+2
 - \frac{k_\alpha}{2}\left( \frac{n-k_\alpha}{\alpha}\right)
\left(\frac{n-k_\alpha}{\alpha}-1\right)
 - \frac{\alpha-k_\alpha}{2}\left( \frac{n-k_\alpha}{\alpha}-1\right)
 \left( \frac{n-k_\alpha}{\alpha}-2\right)\\
 & =  (n+1-\alpha)(n-1)+2
- \frac{k_\alpha}{2}\left(\frac{n-k_\alpha}{\alpha}-1\right)\left(\frac{n-k_\alpha}{\alpha} - \frac{n-k_\alpha}{\alpha}+2\right)
  -\frac{\alpha}{2}
\left( \frac{n-k_\alpha}{\alpha}-1\right)
 \left( \frac{n-k_\alpha}{\alpha}-2\right)\\
 & =  (n+1-\alpha)(n-1)+2
- k_\alpha\left(\frac{n-k_\alpha}{\alpha}-1\right)
  -\frac{\alpha}{2}
\left( \frac{n-k_\alpha}{\alpha}-1\right)
 \left( \frac{n-k_\alpha}{\alpha}-2\right)\\
 & =   (n+1-\alpha)(n-1)+2   - \left(\frac{n-k_\alpha}{\alpha}-1\right)
\left(k_\alpha + \frac{n-k_\alpha}{2}-\alpha\right)\\
& =  (n+1-\alpha)(n-1)+2
-\frac{(n-k_\alpha-\alpha)(n+k_\alpha-2\alpha)}{2\alpha}\\
& =  (n+1-\alpha)(n-1)+2
-\frac{n^2}{2\alpha}+\frac{3n}{2} + \frac{k_\alpha^2}{2\alpha}-\frac{k_\alpha}{2}-\alpha.
\end{align*}
\end{small}
Hence, by abuse of notation, if we let $0\le k_{\alpha}\le \alpha-1$ we have that
$$
\left|\cL_1\left(\bfx^{(\alpha)}\right)\right| = (n+1-\alpha)(n-1)+2
-\frac{n^2}{2\alpha}+\frac{3n}{2} + \frac{k_\alpha^2}{2\alpha}-\frac{k_\alpha}{2}-\alpha,
$$
which implies that
\begin{align*}
\mathsf{diff} & = (n+1-\alpha)(n-1)+2
-\frac{n^2}{2\alpha}+\frac{3n}{2} + \frac{k_\alpha^2}{2\alpha}-\frac{k_\alpha}{2}-\alpha \\
& \ \ \  - (n+1-(\alpha-1))(n-1)-2
     +\frac{n^2}{2(\alpha-1)}-\frac{3n}{2} - \frac{k_{\alpha-1}^2}{2(\alpha-1)}+\frac{k_{\alpha-1}}{2}+\alpha-1\\
     & = -n + n^2\left(\frac{1}{2(\alpha-1)}-\frac{1}{2\alpha}\right) + \left(\frac{k_\alpha^2}{2\alpha}-\frac{k_\alpha}{2}\right) +  \left(\frac{k_{\alpha-1}}{2}-\frac{k_{\alpha-1}^2}{2(\alpha-1)}\right)\\
     & = \frac{n^2}{2\alpha(\alpha-1)}-n + \left(\frac{k_\alpha^2}{2\alpha}-\frac{k_\alpha}{2}\right) +  \left(\frac{k_{\alpha-1}}{2}-\frac{k_{\alpha-1}^2}{2(\alpha-1)}\right).
\end{align*}
Let us consider the following distinct cases for the value of $n$.\\
\textbf{Case $\bf 1$ - } If $n=2\alpha(\alpha-1)$ then $k_\alpha=k_{\alpha-1}=0$ and
$$
\mathsf{diff} =  \frac{(2\alpha(\alpha-1))^2}{2(\alpha-1)\alpha} - 2\alpha(\alpha-1) = 0.$$
\textbf{Case $\bf 2$ - } If $n=2\alpha(\alpha-1)+k$ for some integer $1\le k\le \alpha-2$ then we have that $k_\alpha=k_{\alpha-1}=k$ and
\begin{align*}
\mathsf{diff}  & = \frac{n^2}{2\alpha(\alpha-1)}-n + \left(\frac{k^2}{2\alpha}-\frac{k}{2}\right) +  \left(\frac{k}{2}-\frac{k^2}{2(\alpha-1)}\right)\\
& =  \frac{n^2}{2\alpha(\alpha-1)}-n + \frac{k^2}{2\alpha}-\frac{k^2}{2(\alpha-1)} \\
& =   \frac{n^2}{2\alpha(\alpha-1)}-n - \frac{k^2}{2\alpha(\alpha-1)} \\
& = \frac{(2\alpha(\alpha-1)+k)^2}{2\alpha(\alpha-1)} - 2\alpha(\alpha-1)-k - \frac{k^2}{2\alpha(\alpha-1)}\\
& = 2\alpha(\alpha-1)+ 2k + \frac{k^2}{2\alpha(\alpha-1)}- 2\alpha(\alpha-1)-k - \frac{k^2}{2\alpha(\alpha-1)}\\
& = k > 0.
\end{align*}
\textbf{Case $\bf 3$ - } If $n\ge 2\alpha(\alpha-1)+\alpha - 1$, then we first note that for any $0\le k_{\alpha-1}\le \alpha-2$ we have that $  \left(\frac{k_{\alpha-1}}{2}-\frac{k_{\alpha-1}^2}{2(\alpha-1)}\right)\ge 0$ and hence
\begin{align*}
\mathsf{diff}  & =\frac{n^2}{2\alpha(\alpha-1)}-n + \left(\frac{k_\alpha^2}{2\alpha}-\frac{k_\alpha}{2}\right) +  \left(\frac{k_{\alpha-1}}{2}-\frac{k_{\alpha-1}^2}{2(\alpha-1)}\right)\\
& \ge \frac{n^2}{2\alpha(\alpha-1)}-n + \left(\frac{k_\alpha^2}{2\alpha}-\frac{k_\alpha}{2}\right).
\end{align*}
Define $f:[0,\alpha-1]\to\R$ by $f(x)\triangleq \frac{x^2}{2\alpha} - \frac{x}{2}$. It is easy to verify that $f$ has a single minimum point at $x=\frac{\alpha}{2}$. Hence,
\begin{align*}
 \frac{k_\alpha^2}{2\alpha} - \frac{k_\alpha}{2} = f(k_\alpha) \ge f\left(\frac{\alpha}{2}\right) = \frac{\alpha^2}{4\cdot2\alpha} - \frac{\alpha}{2} = -\frac{\alpha}{8}
\end{align*}
and
$$
\mathsf{diff}\ge  \frac{n^2}{2\alpha(\alpha-1)} - n - \frac{\alpha}{8}.
$$
It holds that
$$
\frac{n^2}{2\alpha(\alpha-1)} - n - \frac{\alpha}{8} \ge 0
$$
if and only if
\begin{align*}
n & \ge \alpha(\alpha-1) + \frac{1}{2}\sqrt{4\alpha^4-7\alpha^3+3\alpha^2} = \alpha(\alpha-1) + \sqrt{\alpha^4 - \frac{7}{4}\alpha^3+\frac{3}{4}\alpha^2}\\
& = \alpha(\alpha-1) + \alpha\sqrt{\alpha^2-\frac{7}{4}\alpha + \frac{3}{4}} = \alpha(\alpha-1) + \alpha\sqrt{\left(\alpha-\frac{3}{4}\right)\left(\alpha-1\right)} .
\end{align*}
Note that
\[
\alpha\sqrt{\left(\alpha-\frac{3}{4}\right)\left(\alpha-1\right)} < \alpha\left(\alpha-\frac{3}{4}\right),
\]
and additionally, it can be verified that for any $\alpha>1$,
$$2\alpha(\alpha-1) + (\alpha-1)\ge  \alpha(\alpha-1) + \alpha\left(\alpha-\frac{3}{4}\right),$$
and thus,
$$n \ge 2\alpha(\alpha-1) + (\alpha-1 )>  \alpha(\alpha-1) + \alpha\sqrt{\left(\alpha-\frac{3}{4}\right)\left(\alpha-1\right)}, $$
which implies that $\mathsf{diff}>0$.\\
\textbf{Case $\bf 4$ - } If $n=2\alpha(\alpha-1)-k$ for some integer $1\le k\le \alpha-2$ then we have that $k_\alpha=\alpha-k,\ k_{\alpha-1}=\alpha-1-k$, and thus
\begin{align*}
\mathsf{diff}  & =\frac{n^2}{2\alpha(\alpha-1)}-n + \left(\frac{(\alpha-k)^2}{2\alpha}-\frac{\alpha-k}{2}\right) +  \left(\frac{\alpha-1-k}{2}-\frac{(\alpha-1-k)^2}{2(\alpha-1)}\right)\\
& = \frac{n^2}{2\alpha(\alpha-1)}-n  + \frac{\alpha^2}{2\alpha} - \frac{2\alpha k}{2\alpha} + \frac{k^2}{2\alpha} + \frac{\alpha-1-k-\alpha+k}{2} - \frac{(\alpha-1)^2}{2(\alpha-1)} + \frac{2(\alpha-1)k}{2(\alpha-1)} - \frac{k^2}{2(\alpha-1)}\\
& = \frac{n^2}{2\alpha(\alpha-1)}-n  + \frac{\alpha}{2} - k + \frac{k^2}{2\alpha} - \frac{1}{2} - \frac{\alpha-1}{2} + k - \frac{k^2}{2(\alpha-1)}\\
& = \frac{n^2}{2\alpha(\alpha-1)}-n  + \frac{k^2}{2\alpha}  - \frac{k^2}{2(\alpha-1)}\\
& =  \frac{(2\alpha(\alpha-1)-k)^2}{2\alpha(\alpha-1)} - 2\alpha(\alpha-1) + k  + \frac{k^2}{2\alpha}  - \frac{k^2}{2(\alpha-1)}\\
& = 2\alpha(\alpha-1) - 2k + \frac{k^2}{2\alpha(\alpha-1)}- 2\alpha(\alpha-1) + k    - \frac{k^2}{2\alpha(\alpha-1)}\\
& = -k < 0.
\end{align*}
\textbf{Case $\bf 5$ - } If $\alpha\le n\le 2\alpha(\alpha-1) - (\alpha-1)$ then we first note that for any $0\le k_{\alpha}\le \alpha-1$ we have that $  \left(\frac{k_{\alpha}^2}{2\alpha}-\frac{k_{\alpha}}{2}\right)\le 0$ and hence
\begin{align*}
\mathsf{diff}  & =\frac{n^2}{2\alpha(\alpha-1)}-n + \left(\frac{k_\alpha^2}{2\alpha}-\frac{k_\alpha}{2}\right) +  \left(\frac{k_{\alpha-1}}{2}-\frac{k_{\alpha-1}^2}{2(\alpha-1)}\right)\\
& \le \frac{n^2}{2\alpha(\alpha-1)}-n + \left(\frac{k_{\alpha-1}}{2}-\frac{k_{\alpha-1}^2}{2(\alpha-1)}\right).
\end{align*}
Define $f:[0,\alpha-2]\to\R$ by $f(x)\triangleq \frac{x}{2} - \frac{x^2}{2(\alpha-1)}$. It can be verified that $f$ has a single maximum point at $x=\frac{\alpha-1}{2}$ and hence
$$
\frac{k_{\alpha-1}}{2}-\frac{k_{\alpha-1}^2}{2(\alpha-1)} = f(k_{\alpha-1})\le f\left(\frac{\alpha-1}{2}\right) = \frac{\alpha-1}{4} - \frac{(\alpha-1)^2}{8(\alpha-1)} =  \frac{\alpha-1}{8}
$$
and
$$
\mathsf{diff} \le \frac{n^2}{2\alpha(\alpha-1)}-n + \frac{\alpha-1}{8}.
$$
It holds that $\frac{n^2}{2\alpha(\alpha-1)}-n + \frac{\alpha-1}{8}\le 0$ if and only if
$$
\alpha(\alpha-1) - \sqrt{\frac{(\alpha-1)^2\alpha(4\alpha-1)}{4}} \le n \le \alpha(\alpha-1) + \sqrt{\frac{(\alpha-1)^2\alpha(4\alpha-1)}{4}}.
$$
Note that
\begin{align*}
\alpha(\alpha-1) - \sqrt{\frac{(\alpha-1)^2\alpha(4\alpha-1)}{4}} & = (\alpha-1)\left(\alpha-\sqrt{\frac{\alpha\left(4\alpha-1\right)}{4}}\right)\\
& =  (\alpha-1)\left(\alpha - \sqrt{\alpha\left(\alpha-\frac{1}{4}\right)}\right) \\
& \le (\alpha-1)\left(\alpha-\left(\alpha-\frac{1}{4}\right)\right)\\
& \le \frac{\alpha-1}{4},
\end{align*}
and
\begin{align*}
\alpha(\alpha-1) + \sqrt{\frac{(\alpha-1)^2\alpha(4\alpha-1)}{4}} & =  (\alpha-1)\left(\alpha+\sqrt{\frac{\alpha\left(4\alpha-1\right)}{4}}\right)\\
& = (\alpha-1)\left(\alpha + \sqrt{\alpha\left(\alpha-\frac{1}{4}\right)}\right)\\
& \ge  (\alpha-1)\left(\alpha+\left(\alpha-\frac{1}{4}\right)\right)\\
& = 2\alpha(\alpha-1) - \frac{\alpha-1}{4},
\end{align*}
and since $\alpha\le n\le 2\alpha(\alpha-1)-(\alpha-1)$, we have that $\mathsf{diff}<0$ as required.
\\

Since $\alpha$ is the number of alternating segments in a sequence of length $n$, it holds that $n\ge \alpha$. In addition, Case $1$ states that for $n=2\alpha(\alpha-1)$ we have that $\mathsf{diff} = 0$. Furthermore, by combining the results from Case $2$ and Case $3$ we have that for any $n > 2\alpha(\alpha-1)$ the value of $\mathsf{diff}$ is a positive number. Similarly Case $4$ and Case $5$ prove that for any $\alpha\le n <2\alpha(\alpha-1)$ the value of  $\mathsf{diff}$ is negative. Thus,
$$\mathsf{diff} = 0 \iff n\ge 2\alpha(\alpha-1).$$
\end{document}